\documentclass[10pt,prb,onecolumn,showpacs,amssymb,superscriptaddress]{revtex4-2}
\usepackage{fullpage}
\usepackage{amsmath, amssymb, amsthm}
\usepackage{physics}
\usepackage[shortlabels]{enumitem}
\usepackage[dvipsnames]{xcolor}
\usepackage{soul}
\usepackage{thmtools}
\usepackage{thm-restate}
\usepackage{graphicx} 
\usepackage[colorlinks=true,linkcolor=blue,citecolor=blue,urlcolor=blue]{hyperref}

\newtheorem{lemma}{Lemma}

\newtheorem{definition}{Definition}

\newcommand{\supp}{\mathrm{supp}}

\renewcommand{\so}{\text{so}}
\newcommand{\su}{\text{su}}

\newcommand{\SO}{\text{SO}}
\newcommand{\SU}{\text{SU}}

\newcommand\descitem[1]{\item{\bfseries #1}\\}

\begin{document}
\title{Duality between string and computational order in symmetry-enriched topological phases}
\author{Paul Herringer}
\affiliation{Institut f\"ur Theoretische Physik, Leibniz Universit\"at Hannover, Appelstra{\ss}e 2, 30167 Hannover, Germany}
\affiliation{Department of Physics and Astronomy, University of British Columbia, 6224 Agricultural Road
Vancouver, BC V6T 1Z1, Canada}
\author{Vir B. Bulchandani}
\affiliation{Department of
Physics and Astronomy, Rice University, 6100 Main Street
Houston, TX 77005, USA}\thanks{on leave}
\affiliation{Department of Physics, National University of Singapore, Singapore 117542}
\author{Younes Javanmard}
\affiliation{Institut f\"ur Theoretische Physik, Leibniz Universit\"at Hannover, Appelstra{\ss}e 2, 30167 Hannover, Germany}
\author{David T. Stephen}
\thanks{Present Address: Quantinuum, 303 S Technology Ct, Broomfield, CO 80021, USA}
\affiliation{Department of Physics and Center for Theory of Quantum Matter, University of Colorado Boulder, Boulder, CO 80309, USA}
\affiliation{Department of Physics, California Institute of Technology, Pasadena, CA 91125, USA}
\author{Robert Raussendorf}
\affiliation{Institut f\"ur Theoretische Physik, Leibniz Universit\"at Hannover, Appelstra{\ss}e 2, 30167 Hannover, Germany}
\affiliation{Stewart Blusson Quantum Matter Institute, University of British Columbia, Vancouver, BC V6T 1Z4, Canada}
\date{Last updated \today}

\begin{abstract}
    We present the first examples of topological phases of matter with uniform power for measurement-based quantum computation. This is possible thanks to a new framework for analyzing the computational properties of phases of matter that is more general than previous constructions, which were limited to short-range entangled phases in one dimension. We show that ground states of the toric code in an anisotropic magnetic field yield a natural, albeit non-computationally-universal, application of our framework. We then present a new model with topological order whose ground states are universal resources for MBQC. Both topological models are enriched by subsystem symmetries, and these symmetries protect their computational power. Our framework greatly expands the range of physical models that can be analyzed from the computational perspective.
\end{abstract}

\maketitle

\section{Introduction}

The classification of quantum phases of matter separates the ground states of many-body quantum systems
into two broad classes: topologically ordered states with ``long-range entanglement’’ and topologically trivial states with ``short-range entanglement’’ \cite{Chen2010}. More concretely, topologically ordered states cannot be created from topologically trivial states by local unitary operations~\cite{Bravyi_2006}. The reverse is also true: topological order cannot be destroyed by local perturbations. 

In the presence of symmetry, even topologically trivial states can be separated into distinct phases of matter; these are called symmetry-protected topological (SPT) phases \cite{Pollmann2010,Chen2011,Schuch2011,Chen2013}. SPT order is similar to ``true’’ topological order in the sense that two states in different SPT phases cannot be related by local unitary operations as long as the symmetry is preserved. The symmetries that define SPT phases can be global, i.e. the symmetry consists of the same operator acting on every site, or subsystem symmetries for which a tensor product of local operators acts on a rigid subset of sites \cite{You2018,Devakul2019fractal}. An example of the latter would be a line-like symmetry that acts on horizontal or vertical lines of sites. Either global or subsystem symmetry can similarly be imposed upon topologically ordered states, giving rise to so-called symmetry-enriched topological (SET) order~\cite{Chen2017,Barkeshli2019,Stephen2020,Stephen2022,hsin2024subsystem}. 

Over the last decade it has been shown that an \emph{a priori} distinct property emerges in certain quantum phases of matter, namely computational power. This power is revealed by measurement-based quantum computation (MBQC), which is a scheme for implementing quantum computations via adaptive, site-local measurements on a many-body quantum state \cite{Raussendorf2003}. In MBQC, the allowed measurement operations are always the same, so that the quantum state being measured defines the underlying computational resource and determines which logical operations can be performed. More formally, for every pure quantum state, one can derive the the Lie group of quantum circuits that are realizable through MBQC. In some cases, this Lie group is uniform over a given phase of matter. For example, the SPT phase with $\mathbb{Z}_2 \times \mathbb{Z}_2$ symmetry around the 1D cluster state---the Haldane phase---uniformly gives rise to the group $SU(2)$ of quantum circuits \cite{Miyake2010,Else2012,Stephen2017,Raussendorf2023measurementbased}. There are even examples of phases of matter in 1D and 2D SPT systems with universal computational power \cite{Raussendorf2019,Stephen2019subsystem,Devakul2018universal,Daniel2020computational,Stephen2024}. Thus we can speak of computational order in addition to more ``physical'' notions of order for many-body quantum systems and similarly of ``computational phases of matter'', which are the main focus of this work.

Until now, such computational phases of matter had only been identified in systems with SPT order \cite{Miyake2010,Else2012,Stephen2017,Raussendorf2023measurementbased,Raussendorf2019,Stephen2019subsystem,Devakul2018universal,Daniel2020computational,Stephen2024}. One might thus conjecture that computational order is simply SPT order viewed through the lens of MBQC. However, as we show in this paper, this is not the case. Specifically, we present two examples of computational phases in long-range entangled topologically ordered systems with symmetry, i.e., SET order. Our first example is the familiar toric code perturbed by an anisotropic magnetic field. It is known that the toric code state has limited but non-trivial computational power in MBQC \cite{Bravyi2007}, and we show that this power extends to a symmetry-enriched phase around it. Our second example yields a computationally universal SET phase, demonstrating that computational universality is compatible with SET order as well as SPT order. Our findings thus raise the following question: {\em{What is the precise relationship between computational order and physical order?}}





\subsection{Overview of results}\label{sec:summary}

Our main result, Theorem~\ref{thm:power}, is a generalization of the existing 1D computational order formalism~\cite{Raussendorf2023measurementbased} to 2D and to states that may have topological order. We apply this theorem to show that both the toric code and our new stabilizer model exhibit non-trivial computational order, in a manner that is inextricably related to their unusual (symmetry-enriched) topological order.
For comparison, we first state a counterpart of Theorem~\ref{thm:power} for SPT order in 1D. It is the informal (as well as restricted and partial) version of Corollary~1 in \cite{Raussendorf2023measurementbased}. The main idea is that computational order is characterized by the same string order parameters $\sigma_k$ that detect SPT order in 1D.

\begin{restatable}[1D computational power, informal \cite{Raussendorf2023measurementbased}]{thm}{power} \label{thm: power informal 1d}
Consider a short-range entangled state $|\Phi\rangle$ of a spin-1/2 chain, symmetric under a group $G=(\mathbb{Z}_2)^m$. Furthermore, assume that the string order parameters $\sigma_k$ are non-zero for all sites $k$. Then, the group of realizable gates is $\mathcal{L}=\exp (-i \mathcal{A})$, where $\mathcal{A}$ is the Lie algebra generated by $\{T(g_k)\}$, with $T(G)$ a projective representation of the symmetry group $G$, and $g_k\in G$ a special group element associated with the $k$-th spin through the construction of the order parameter $\sigma_k$.
\end{restatable}

For illustration, let us consider the example of the cluster phase~\cite{Raussendorf2023measurementbased}, with symmetry group $G=\mathbb{Z}_2\times \mathbb{Z}_2$. The symmetries and operators yielding the order parameters are depicted in Fig.~\ref{ClusterPhase}. In this case, the projective representation $T(G)$ is the maximally non-commutative one -- $T(00)=I$, $T(01)=X$, $T(10)=Z$, $T(11) = Y$. The special group elements associated with the individual sites are $(01)$ for all even and $(10)$ for all odd sites. The Lie algebra ${\cal{A}}$ is thus generated by the Pauli operators $Z$ and $X$, hence the resulting Lie group of realizable circuits is ${\cal{L}}=SU(2)$.

\begin{figure}
    \begin{center}
    \includegraphics[width=0.95\linewidth]{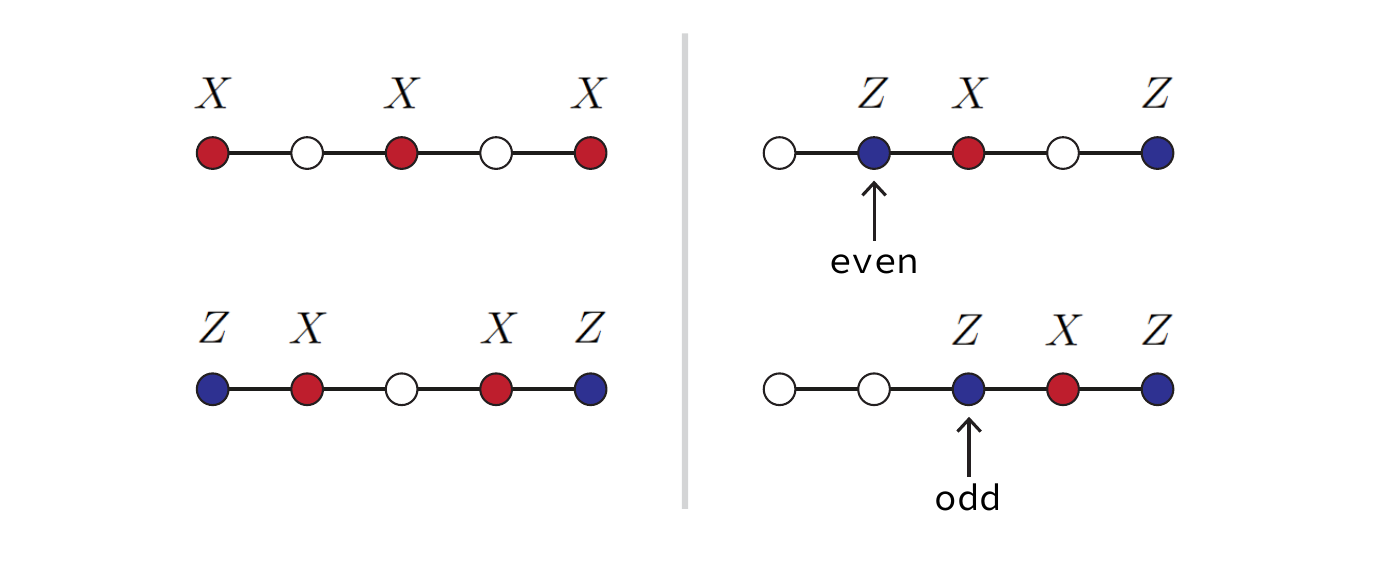}
    \caption{\label{ClusterPhase}Symmetry action and order parameters in the $\mathbb{Z}_2\times \mathbb{Z}_2$ cluster phase in 1D. \textbf{Left panel}: Action of the symmetry on the physical spins. Top: representation of the group element $(01)\in \mathbb{Z}_2\times \mathbb{Z}_2$, bottom: representation for the group element $(10)$. \textbf{Right panel}: Order parameters ending in even and odd sites, respectively.}
    \end{center}
\end{figure}

\subsubsection{Main theorem}

We extend Theorem \ref{thm: power informal 1d} as follows.
\begin{restatable}[Computational power, informal]{thm}{power} \label{thm:power informal 2d}
    Consider a state $\ket{\Phi}$ of a 2D spin-1/2 system, such that
    \begin{enumerate}[(a)]
        \item $\ket{\Phi}$ is symmetric under a group $G\cong\mathbb{Z}_2^m$. 
        \item $\ket{\Phi}$ \textbf{is related to a suitable stabilizer state by a local, finite-depth, symmetric unitary.} 
        \item The order parameters $\sigma_k$ are non-zero for every site $k$. 
    \end{enumerate}
    Then the Lie group of quantum gates realizable by MBQC with $\ket{\Phi}$ as a resource state is $\mathcal{L}=\exp{-i\mathcal{A}}$, where $\mathcal{A}$ is the Lie algebra generated by $\{T(g_k)\}$ with $T(G)$ a projective representation of $G$ and $g_k$ a special group element associated with site $k$ through the construction of the order parameter $\sigma_k$.
\end{restatable}
We emphasize condition (b) of Theorem \ref{thm:power informal 2d} because it allows us to address certain phases of matter that do not contain product states, thereby unlocking the study of computational phases of matter beyond invertible, short-range entangled phases. In particular, this includes phases of matter with topological order. 

Some comments on Theorem~\ref{thm:power informal 2d} are in order. First, the presence of a symmetry group $G$ means that, when we want to apply the Theorem to states with topological order, we necessarily enter the realm of symmetry-enriched topological (SET) phases of matter. Then, the Theorem says that certain SET phases of matter have uniform computational power, akin to the previous results for SPT phases \cite{Else2012,Stephen2017,Raussendorf2023measurementbased}.

Second, as in the case of short-range entangled phases, we find that the most interesting SET phases from the perspective of MBQC have subsystem symmetries \cite{Raussendorf2019,Stephen2019subsystem,Devakul2018universal,Daniel2020computational}, which are symmetries that act on rigid lower-dimensional submanifolds. For the examples we consider, these subsystem symmetries are line-like, and the order parameters $\sigma_k$ appearing in Theorem~\ref{thm:power informal 2d} are string-order parameters that descend from the line-like symmetries. Topological phases enriched by subsystem symmetries, dubbed subsystem SET (SSET) phases, have recently been studied in Refs.~\cite{Stephen2022,hsin2024subsystem}. While these phases of matter may appear to be relatively obscure, a simple realization of this physics arises upon perturbing the toric code by a bond-anisotropic magnetic field, as we now discuss.


\begin{figure}[t]
    \centering
   \includegraphics[width=\linewidth]{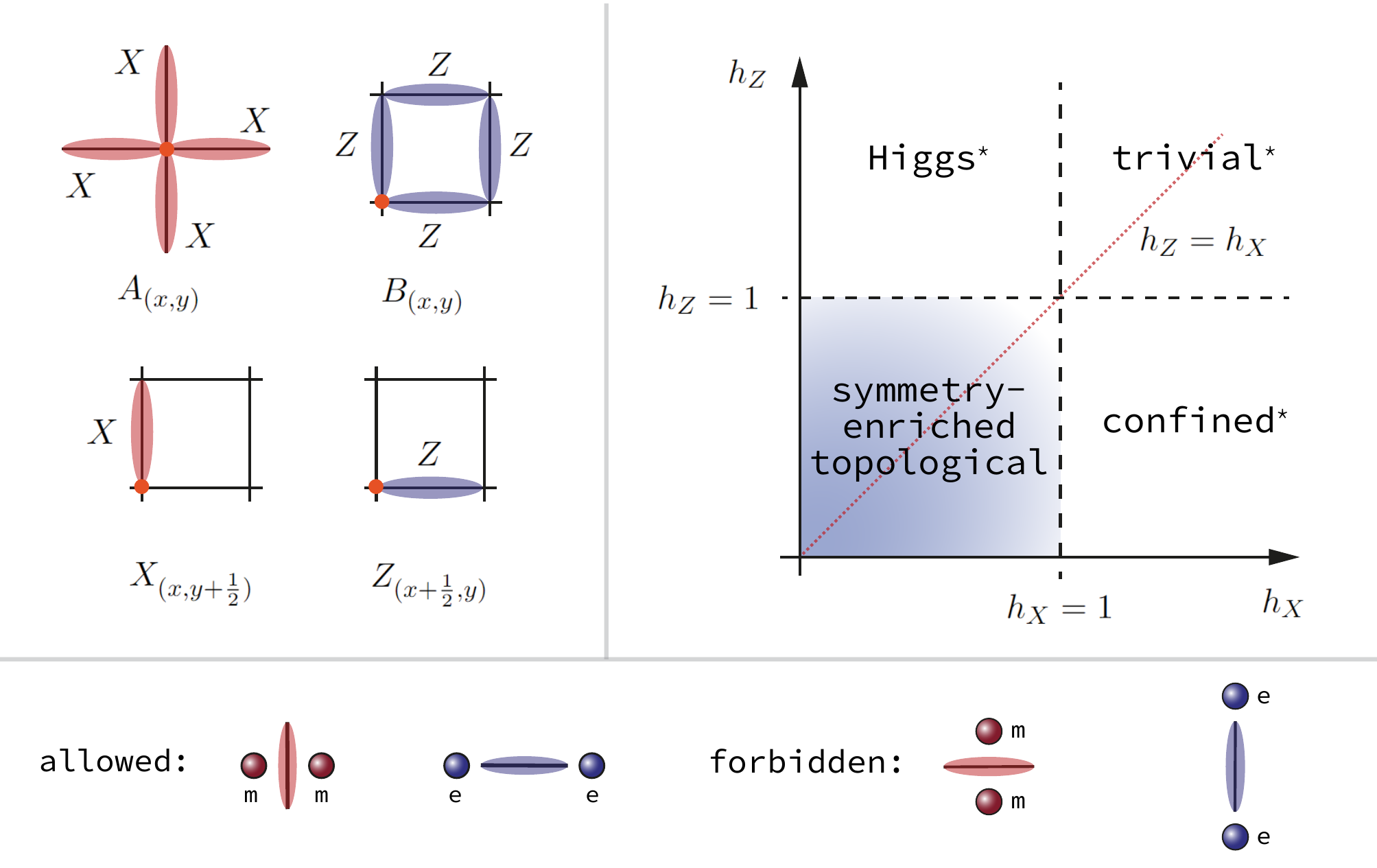}
    \caption{
    \textbf{Top left panel}: Our labelling conventions for the terms in the anisotropically perturbed toric code Hamiltonian Eq. \eqref{eq:TCH}. Orange dots indicate the lattice point $(x,y)$ and throughout this paper, red and blue lobes will correspond to $X$ and $Z$ Pauli operators respectively. \textbf{Bottom panel}: Unlike the usual magnetic-field perturbation of the toric code, the anisotropic magnetic-field perturbation Eq. \eqref{eq:dH} only creates perfectly horizontal pairs of electric or magnetic anyons. Thus only a subset of the usual pair-creation processes are allowed. \textbf{Top right panel:} The ground-state phase diagram of the resulting Hamiltonian. Asterisks indicate that charges are only confined in the vertical direction or equivalently the presence of subsystem symmetry as discussed in Section \ref{sec:sset}. We follow the standard terminology~\cite{TCPD} borrowed from $U(1)$ gauge theory, according to which the $e$-condensed phase $h_Z>1$ yields an $m$-confined or ``Higgs'' phase and the $m$-condensed phase $h_X>1$ yields an $e$-confined or simply ``confined'' phase. The line $h_Z = h_X$ is indicated in red.}
    \label{fig:anyons}
\end{figure}

\subsubsection{First example: toric code in an anisotropic field}
As an illustrative example of a model that falls under the assumptions of Theorem 1, we consider the toric code~\cite{kitaev2003fault} 
\begin{equation}
    H_0 = -\sum_{(x,y)\in\Lambda} A_{(x,y)} -\sum_{(x,y)\in \Lambda} B_{(x,y)}
\end{equation}
on a periodic square lattice $\Lambda$ of $L_x$ by $L_y$ sites, where $(x,y)$ labels the Cartesian coordinates of each vertex with $x=1,2,\ldots,L_x$ and $y=1,2,\ldots,L_y$ and $A_{(x,y)}$ and $B_{(x,y)}$ are the usual star and plaquette operators, see Fig.~\ref{fig:anyons}. We perturb $H_0$ by a bond-anisotropic magnetic field,
\begin{equation}
\label{eq:dH}
\delta H = - h_X \sum_{(x,y)\in\Lambda} X_{(x,y+\frac12)} - h_Z \sum_{(x,y)\in\Lambda} Z_{(x+\frac12,y)}
\end{equation}
where we use half-integer coordinates to label edges and $h_X, \, h_Z \geq 0$. This yields the perturbed toric code Hamiltonian
\begin{equation}
\label{eq:TCH}
H = H_0 + \delta H.
\end{equation}
The phase diagram of this model is shown in Fig.~\ref{fig:anyons}. In contrast to more standard magnetic field perturbations without bond anisotropy~\cite{TCPD}, the anisotropic form of the perturbation Eq. \eqref{eq:dH} guarantees that the parity of the number of anyons (excitations of $A_{(x,y)}$ or $B_{(x,y)}$) is conserved along every row, see Fig.~\ref{fig:anyons}. These conservation laws can be expressed as subsystem symmetries,
\begin{equation} \label{eq:subsystem_symmetries}
    U(g_y^e) = \prod_{x=1}^{L_x} A_{(x,y)},\qquad U(g_y^m) = \prod_{x=1}^{L_x} B_{(x,y)},
\end{equation}
which are equal to products of $X$ operators or products of $Z$ operators respectively along two parallel rows, see Fig.~\ref{fig:subsystem_symms}. The phase diagram in Fig.~\ref{fig:anyons} should be thought of as a deformation of the toric code that preserves these subsystem symmetries. In fact, we will find that the unperturbed toric code state lies in a non-trivial SSET phase of matter protected by these symmetries.

Using these symmetries, we can define string order parameters,
\begin{equation} \label{eq:sops_pbc}
\begin{aligned}
    S_{[a,b]}(g_y^e) &= \prod_{x=a}^b A_{(x,y)} =  X_{(a-\frac12,y)}\left(\prod_{x=a}^{b} X_{(x,y+\frac12)}X_{(x,y-\frac12)}\right)X_{(b+\frac12,y)},\\
    S_{[a,b]}(g_y^m) &=  \prod_{x=a}^b B_{(x,y)} = Z_{(a,y+\frac12)}\left(\prod_{x=a}^{b} Z_{(x+\frac12,y)}Z_{(x+\frac12,y+1)}\right)Z_{(b+1,y+\frac12)},
\end{aligned}
\end{equation}
as pictured in Fig.~\ref{fig:subsystem_symms}. All ground states of $H_0$ have perfect string order, $\langle S_{[a,b]}(g_y^{e/m})\rangle = 1$. When we perturb away from this point, the string order parameters are no longer perfect, but crucially their expectation values do not decay with their length $|a-b|$.
These string order parameters (or rather their appropriate generalizations for open boundary conditions) are precisely the computational order parameters that appear in Theorem~\ref{thm:power informal 2d}. By mapping Eq.~\ref{eq:TCH} to a stack of 1D transverse-field Ising models (see e.g. Refs. \cite{WenPlaquette} and \cite{wachtel2019confinement} for related results) we prove that these order parameters are non-zero everywhere in the symmetry-enriched topological phase $h_X,\,h_Z<1$, so that every ground state in this phase is an MBQC resource for the same family of computations described by Theorem~\ref{thm:power informal 2d}. This family of computations turns out to be non-universal since the Lie group of gates that can be implemented in MBQC maps onto classically simulable free-fermionic circuits. This shortcoming is remedied by our second example.

\begin{figure}
    \centering
    \includegraphics[width=\linewidth]{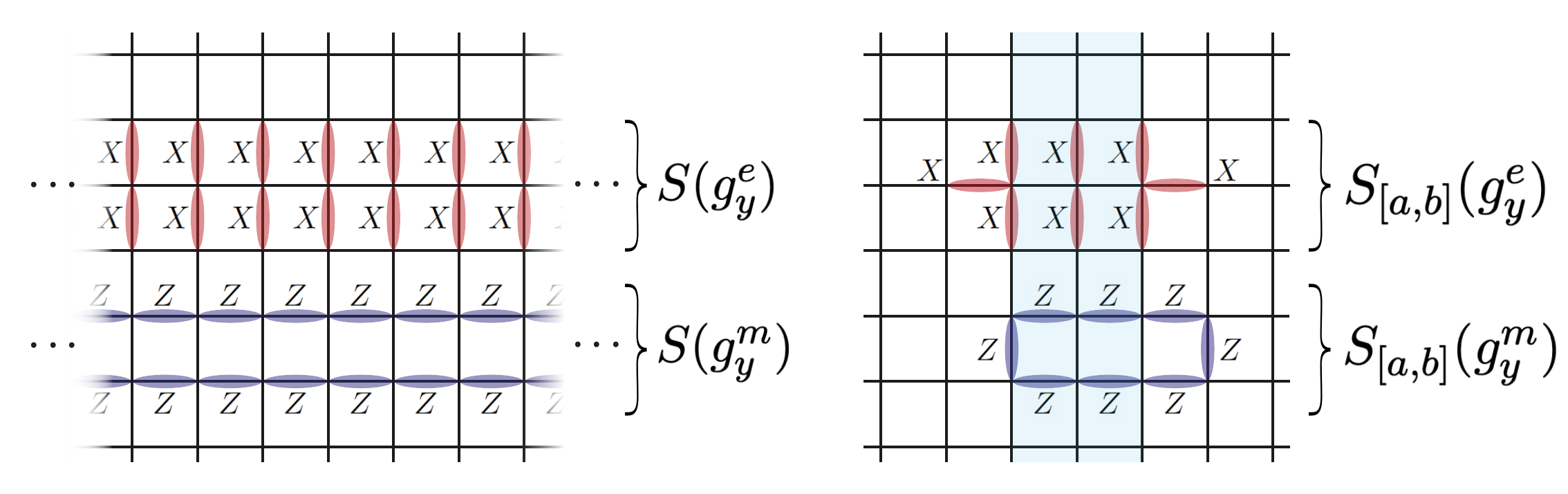}
    \caption{
    \textbf{Left}: The subsystem symmetries of the toric code in an anisotropic magnetic field. \textbf{Right}: The corresponding string order parameters. The highlighted region corresponds to the interval $x\in[a,b]$.}
    \label{fig:subsystem_symms}
\end{figure}

\subsubsection{Second example: a computationally universal topological phase}
We propose a new stabilizer model that yields universal resource states for MBQC,
\begin{equation}
    H_{*} = -\sum_x\sum_{y\text{ even}} C_{x,y} -\sum_x\sum_{y\text{ odd}} D_{x,y}
\end{equation}
where
\begin{equation}
    C_{x,y} = X_{x-1,y}X_{x,y}X_{x+1,y}X_{x,y-1}X_{x,y+1}, \quad D_{x,y} = Z_{x-1,y}Z_{x,y}Z_{x+1,y}Z_{x,y-1}Z_{x,y+1}.
\end{equation}
This model again has topological order along with a set of subsystem symmetries and corresponding order parameters that satisfy the conditions of Theorem.~\ref{thm:power informal 2d}, and we find that the corresponding Lie group of quantum gates in universal. This model and its symmetries therefore define a computationally universal SSET phase of matter. Interestingly, we also show that this model has the same topological order and symmetry fractionalization as a stack of two toric codes, which do not enable universal MBQC. This apparent tension is resolved by the fact that the two models have differently structured subsystem symmetries, highlighting the importance of the microscopic symmetry representation for MBQC.

\subsection{Structure of paper}

The rest of the paper is structured as follows. In Section \ref{sec:comp order} we discuss our MBQC scheme for 2D systems of qubits and formally state Theorem \ref{thm:power}. Section \ref{sec:tc example} is devoted to the toric code example. In Section \ref{sec:tc comp a} we apply Theorem~\ref{thm:power} to the toric code model with an anisotropic magnetic field and demonstrate that this model has computational order. Then we discuss the topological characteristics and subsystem symmetries of the model in Section \ref{sec:tc topo}. We conclude our study of the toric code example by relating computational order and topological (string) order in Section \ref{sec:tc relations}.

In Section \ref{sec:xz star} we introduce the $XZ$-star model, discuss its properties from the perspective of topological order and subsystem symmetries, and prove that its ground states are computationally universal for MBQC. We conclude with a general discussion of our results in Section \ref{sec:disc}.

\section{Computational order}
\label{sec:comp order}
In this section we characterize the computational power of 2D systems of qubits with symmetries. Specifically, we consider states that are symmetric under a group $G \cong \mathbb{Z}_2^k$, and computational power in the context of MBQC. MBQC leverages single-qubit measurements on a many-body quantum state to perform computations; the computational power of a state refers to the number of logical qubits available for computation and the Lie group of logical gates that can be implemented via measurement. 

Here we show that for symmetric quantum states these computational properties are characterized by a set of computational order parameters (COPs).
Our order parameters and MBQC scheme are a generalization of the scheme defined in~\cite{Raussendorf2023measurementbased}, which uses 1D string order parameters and applies to short-range entangled states.
The technical reason for requiring short-range entangled states is that the proofs require expectation values of local observables to decorrelate, $\expval{PQ} = \expval{P}\expval{Q}$, when their supports are separated by a finite distance. 

In Lemma \ref{lemma:decorr} we prove a less restrictive condition, namely that decorrelation also occurs if the state in question is related to an appropriate stabilizer state by a local, finite-depth unitary circuit. Notably, the allowed stabilizer states include states with topological order, such as our toric code model. After the statement and proof of Lemma \ref{lemma:decorr}, we define the computational order parameters and specify how they characterize computational power for MBQC with Theorem \ref{thm:power}.

\subsection{Decorrelation} \label{sec:decorr}

\begin{figure}
    \centering
    \includegraphics[width=0.6\linewidth]{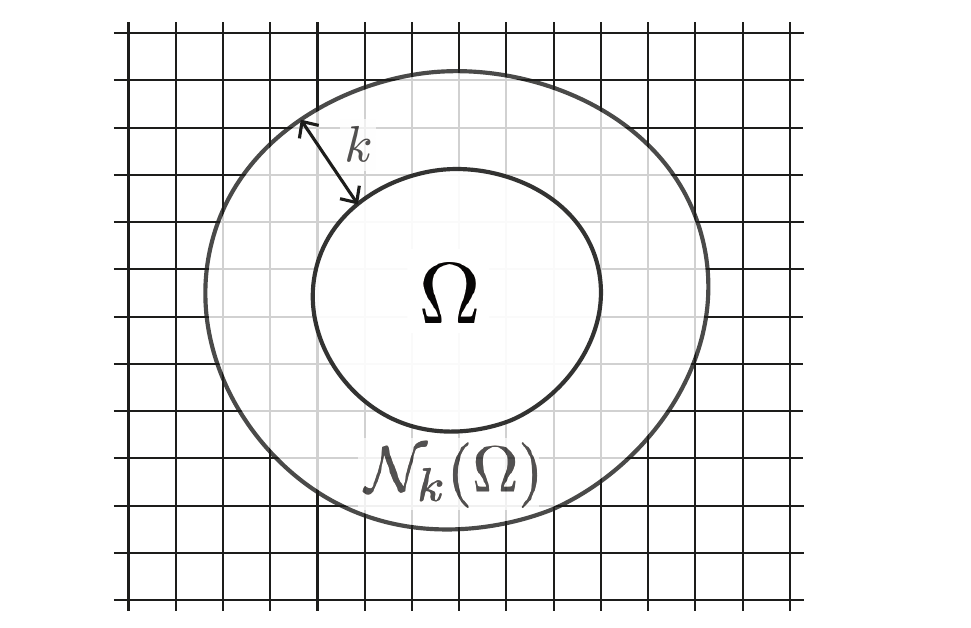}
    \caption{Expanding a region $\Omega$ of the lattice by a ribbon of width $k$ to obtain the $k$-neighborhood $\mathcal{N}_k(\Omega)$.}
    \label{fig:Nk}
\end{figure}

Consider a region $\Omega$ of a 2D lattice of qubits. We require a way to expand $\Omega$ by a fixed amount, so we introduce the $k$-neighborhood of $\Omega$, which is $\Omega$ plus a ribbon of width $k$ around its perimeter (Fig. \ref{fig:Nk}). Formally, we say that the $k$-neighborhood of $\Omega$ is the set of all points $(x,y)$ such that a sphere $\mathcal{S}_k(x,y)$ of radius $k$ centered at $(x,y)$ touches $\Omega$:
\begin{equation}
    \mathcal{N}_k(\Omega) = \{(x, y): \mathcal{S}_k(x, y) \cap \Omega \neq \varnothing \}.
\end{equation}
For an operator $P$ we use the shorthand $\mathcal{N}_k(P)$ to denote $\mathcal{N}_k(\supp(P))$.
\begin{definition}
    Let $W$ a quantum circuit acting on a 2D lattice of spins and $\Omega$ a region of the lattice. The spread $\Delta(W)$ is the smallest non-negative integer such that for any choice of $\Omega$ and any linear operator $P$ whose support is contained within $\Omega$,
    \begin{equation}
        \supp(W^\dagger P W) \subset \mathcal{N}_\Delta(\Omega).
    \end{equation}
\end{definition}
\noindent In other words, if $W$ has spread $\Delta$ then $\supp(W^\dagger P W)$ lies within $\supp(P)$ plus a ribbon of width $\Delta$.

Let $S_{\Phi_0}$ be the stabilizer group of $\ket{\Phi_0}$. Given a region $\Omega$, we can write a generating set for $S_{\Phi_0}$ as follows:
\begin{equation}
    S_{\Phi_0} = \langle S(\Omega), \bar{S}(\Omega), \partial S(\Omega) \rangle, 
\end{equation}
where $S(\Omega)$ are supported within $\Omega$, $\bar{S}(\Omega)$ are supported within the complement of $\Omega$, and $\partial S(\Omega)$ cross the boundary of $\Omega$.

\begin{lemma} \label{lemma:decorr}
    Consider a state $\ket{\Phi}$ created from a stabilizer state $\ket{\Phi_0}$ by a quantum circuit $W$ with spread $\Delta$. Let $P$ and $Q$ two linear operators and suppose we can choose $\partial S(\mathcal{N}_\Delta(P))$ such that none of its elements overlap with the $\Delta$-neighborhood of $Q$: $\supp(s) \cap \mathcal{N}_\Delta(Q) = \varnothing$ for all $s\in \partial S(\mathcal{N}_\Delta(P)).$ Then 
    \begin{equation}
        \bra{\Phi}PQ\ket{\Phi} = \bra{\Phi}P\ket{\Phi}\bra{\Phi}Q\ket{\Phi}.
    \end{equation}
\end{lemma}
\begin{proof}
    Letting $\Pi[S]$ be the projectors onto a set of stabilizer generators, we can write 
    \begin{equation}
        \ket{\Phi_0} \propto \Pi[S(\Omega), \bar{S}(\Omega), \partial S(\Omega)]\ket{\Phi_0}. 
    \end{equation}
    Now consider $W^\dagger P W$, which has support contained in $\mathcal{N}_\Delta(P)$. We can pull a copy of $\Pi[\bar{S}(\mathcal{N}_\Delta(P)]$ past $W^\dagger P W$ because their supports do not overlap:
    \begin{equation} \label{eq:stabs past A}
        (W^\dagger P W) \Pi[S_{\Phi_0}] = \Pi[\bar{S}(\mathcal{N}_\Delta(P)] (W^\dagger A W) \Pi[S_{\Phi_0}]. 
    \end{equation}
    With the assumption of the lemma, we also have
    \begin{equation} \label{eq:stabs past B}
        (W^\dagger Q W) \Pi[S_{\Phi_0}] = \Pi[S(\mathcal{N}_\Delta(P)), \partial S(\mathcal{N}_\Delta(P))] (W^\dagger P W) \Pi[S_{\Phi_0}]. 
    \end{equation}
    Combining the results of \eqref{eq:stabs past A} and \eqref{eq:stabs past B}, we have
    \begin{equation}
        \begin{aligned}
            \bra{\Phi} PQ \ket{\Phi} &= \bra{\Phi_0} \Pi[S_{\Phi_0}] (W^\dagger P W) (W^\dagger Q W) \Pi[S_{\Phi_0}] \ket{\Phi_0} \\ 
            &= \bra{\Phi_0} \Pi[S_{\Phi_0}] (W^\dagger P W) \Pi[\bar{S}(\mathcal{N}_\Delta(P), S(\mathcal{N}_\Delta(P)), \partial S(\mathcal{N}_\Delta(P))] (W^\dagger Q W) \Pi[S_{\Phi_0}] \ket{\Phi_0} \\
            &= \bra{\Phi_0} \Pi[S_{\Phi_0}] (W^\dagger P W) \Pi[S_{\Phi_0}] (W^\dagger Q W) \Pi[S_{\Phi_0}] \ket{\Phi_0} \\
            &= \bra{\Phi_0} \Pi[S_{\Phi_0}] (W^\dagger P W) \Pi[S_{\Phi_0}] \ket{\Phi_0} \bra{\Phi_0} \Pi[S_{\Phi_0}] (W^\dagger Q W) \Pi[S_{\Phi_0}] \ket{\Phi_0} \\
            &= \bra{\Phi}P\ket{\Phi}\bra{\Phi}Q\ket{\Phi}  
        \end{aligned}
    \end{equation}
\end{proof}


As an example, let us apply Lemma \ref{lemma:decorr} to the toric code ground state on a lattice with smooth, open boundaries. For simplicity, let $\supp(P)$ and $\supp(Q)$ be rectangular regions with either rough or smooth boundaries in the bulk. Then we have three cases:
\begin{enumerate}[(i)]
    \descitem{Both smooth boundaries} A rectangular region with a smooth boundary is defined by the edges around its perimeter. The toric code stabilizers that cross this boundary are vertex operators 
    (Fig. \ref{fig:bdaries}). If we have a circuit $W$ with spread $\Delta$, then $\mathcal{N}_\Delta(P)$ and $\mathcal{N}_\Delta(Q)$ are rectangular regions with length and width both increased by $2\Delta$. The set $\partial S(\mathcal{N}_\Delta(P))$ consists of vertex stabilizers that cross the boundary of $\mathcal{N}_\Delta(P)$, and for Lemma \ref{lemma:decorr} to be satisfied we require that none of these overlap with $\mathcal{N}_\Delta(Q)$. This will be the case if there is at least one plaquette separating $\mathcal{N}_\Delta(P)$ from $\mathcal{N}_\Delta(Q)$, or in other words, $\supp(P)$ is separated from $\supp(Q)$ by at least $2\Delta+1$ unit cells (Fig. \ref{fig:bdaries}).
    \descitem{Both rough boundaries} A rectangular region with a rough boundary is defined by the plaquettes around its perimeter. The toric code stabilizers that cross this boundary are plaquette operators 
    (Fig. \ref{fig:bdaries}). This situation is analagous to Case (i) with vertex and plaquette operators swapped, so again we find that Lemma \ref{lemma:decorr} is satisfied if $\supp(P)$ is separated from $\supp(Q)$ by at least $2\Delta+1$ unit cells.
    \descitem{One smooth and one rough boundary} When $\supp(P)$ has a smooth boundary and $\supp(Q)$ has a rough boundary, the stabilizers that cross the boundary of $\mathcal{N}_\Delta(P)$ are vertex operators
    and the stabilizers that cross the boundary of $\mathcal{N}_\Delta(Q)$ are plaquette operators.
    As shown in Fig. \ref{fig:bdaries}, Lemma \ref{lemma:decorr} is satisfied if $\mathcal{N}_\Delta(P)$ and $\mathcal{N}_\Delta(Q)$ are separated by at least 1.5 unit cells, which requires $\supp(P)$ and $\supp(Q)$ to be separated by at least $2\Delta+1.5$ unit cells.
\end{enumerate}

\begin{figure}
    \centering
    \begin{tabular}{ccc}
         \includegraphics[width=0.32\textwidth]{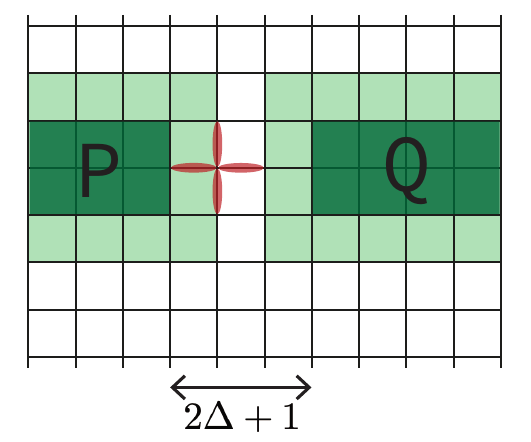} & 
         \includegraphics[width=0.32\textwidth]{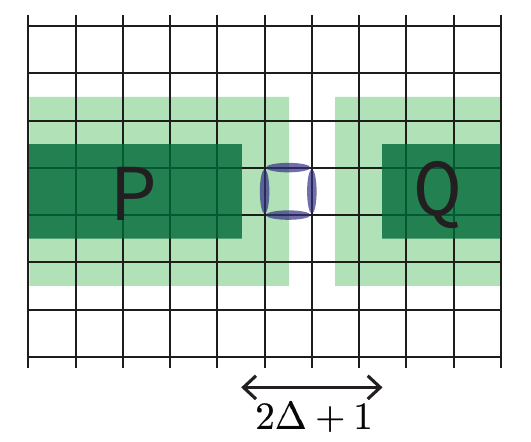} & 
        \includegraphics[width=0.32\textwidth]{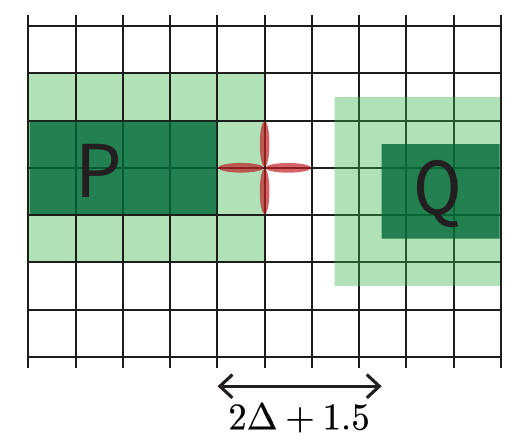} \\
         (a) & (b) & (c)  
    \end{tabular}
    \caption{Separation between operators $P$ and $Q$ required to satisfy Lemma \ref{lemma:decorr} in states related to the toric code state by a local, finite-depth unitary with spread $\Delta$. Darker green regions represent supports of $P$ and $Q$ respectively, while lighter green regions represent the neighbourhoods $N_\Delta(P)$ and $N_\Delta(Q)$. (a) If the supports of $P$ and $Q$ have smooth boundaries, we require vertex stabilizers that cross the boundary of $\mathcal{N}_\Delta(P)$ not to cross the boundary of $\mathcal{N}_\Delta(Q)$, so the separation must be $2\Delta +1$. (b) The case that $P$ and $Q$ are supported on regions with rough boundaries. (c) The case that $\supp(P)$ has a smooth boundary and $\supp(Q)$ has a rough boundary.}
    \label{fig:bdaries}
\end{figure}


\subsection{Symmetry representations}

Our setting for MBQC is a 2D lattice of qubits with open boundaries. We define three subsets of the lattice: the left boundary $L$, the bulk $B$, and the right boundary $R$. Let $\ket{\Phi}$ be a quantum state of the lattice symmetric under a linear representation $U(G)$ of a group $G \cong \mathbb{Z}_2^m$:
\begin{equation} \label{eq:Ug}
    U(g)\ket{\Phi} = (-1)^{\chi(g)}\ket{\Phi} \qquad g\in G,
\end{equation}
where $\chi(g)\in \{0,1\}.$ We assume that $U(G)$ is a tensor product of single-qubit Pauli operators $I, X, Y$ or $Z$ which commute in the bulk but may anticommute at the boundaries. In other words, $U(G)$ is a product of site-local representations which are linear in the bulk and projective at the boundaries. Specifically, we have:
\begin{enumerate}
    \item Linear representations in the bulk: $\hat{u}_{i}(g) \in \{X_i, Y_i, Z_i, I_i\}, i\in B$.
    \item Projective representations on the left boundary: $\hat{v}_{i}^L(g) \in \{X_i, Y_i, Z_i, I_i\}, i\in L$. 
    \item Projective representations on the right boundary: $\hat{v}_{i}^R(g) \in \{X_i, Y_i, Z_i, I_i\}, i\in R$.
\end{enumerate}
We will often refer to symmetry representations acting on the entire left boundary, right boundary, or bulk, which we denote
\begin{equation}
    V_L(g) = \bigotimes_{i\in L} \hat{v}_{i}^L(g), \qquad V_R(g) = \bigotimes_{i\in R} \hat{v}_{i}^R(g), \qquad U_{B}(g) = \bigotimes_{i \in B} \hat{u}_{i}(g).
\end{equation}
With this notation we can expand $U(g)$ as
\begin{equation} \label{eq:global symms from local reps}
        U(g) = V_L(g)\otimes U_{B}(g) \otimes V_R(g).
    \end{equation}
For $U(g)$ to be linear, \eqref{eq:global symms from local reps} requires the commutation relations of $V_L(g)$ to match those of $V_R(g)$. We formalize this by defining a function $\kappa: G \times G \to \mathbb{Z}_2$ via
\begin{equation}
    V_L(g)V_L(g') = (-1)^{\kappa(g,g')}V_L(g')V_L(g)
\end{equation}
and demand that $V_R(g)$ satisfy the same equation, i.e.
\begin{equation}
    V_R(g)V_R(g') = (-1)^{\kappa(g,g')}V_R(g')V_R(g).
\end{equation}

\subsection{Computational order parameters}

Let $k$ be a site in the bulk, and denote by $\{>k\}$ the set of all sites for which the horizontal distance to the right boundary is less than $k$. In other works, $\{>k\}$ is the rectangular region starting immediately to the right of site $k$ and running all the way up to and including the right boundary. 
We denote the restriction of $U(g)$ to $\{>k\}$ by $U_{>k}(g)$.
This operator may fail to commute with $U(G)$ because the projective representation $V_R(G)$ at the right boundary is no longer counterbalanced by $V_L(g)$ at the left boundary. However, we can search for an operator $\hat{\beta}_{k}(g)$ acting on site $k$ such that
\begin{equation}
    R_{k}(g) := \hat{\beta}_{k}(g) \otimes U_{>k}(g)
\end{equation}
commutes with all $U(G)$. When such an operator exists, we say that $U(g)$ is localizable to site $k$, and that $\hat{\beta}_{k}(g)$ localizes $U(g)$. The condition for $U(g)$ to be localizable at site $k$ is that there exists $\hat{\beta}_{k}(g)$ such that
\begin{equation} \label{eq:localizability condition}
    \hat{u}_{k}(g') \hat{\beta}_{k}(g) = (-1)^{\kappa(g,g')} \hat{\beta}_{k}(g)\hat{u}_{k}(g') \qquad \forall\:g'\in G.
\end{equation}
At each site we denote $\mathcal{G}_{k}\subset G$ the subgroup of $G$ for which there is a $\hat{\beta}_{k}(g)$ satisfying \eqref{eq:localizability condition}. The set of all elements of $g$ that can be localized to at least one site is denoted $\mathcal{G} = \bigcup_{i} \mathcal{G}_{i}$.





\begin{lemma}
    If \eqref{eq:localizability condition} has a solution for site $k$ and group element $g_k\in \mathcal{G}_k$, then without loss of generality $\hat{\beta}_k(g_k)$ may be chosen to be a Pauli operator.
\end{lemma}

\begin{proof}
    Without loss of generality assume that $\hat{u}_k(g) \in \{X_k, I_k\}$ for all $g \in G$, then $\hat{\beta}_k(g_k)$ must anticommute with $X_k$. The space of single-qubit operators that anticommute with $X$ is the complex vector space spanned by $\{Y_k, Z_k\}$, so we can make $\hat{\beta}_k(g_k)$ a Pauli operator by choosing $\hat{\beta}_k(g_k)$ to be $Y_k$ or $Z_k$.
\end{proof}

The reason for demanding that $R_{k}(g_k)$ commutes with all $U(G)$ is to allow for non-zero expectation values $\sigma_{k}(g_k) = \expval{R_{k}(g_k)}_\Phi$ on symmetric states $\ket{\Phi}$. These expectation values will turn out to be computational order parameters, and we will show that MBQC is possible whenever these order parameters are non-zero. Note that $\sigma_{k}(g_k)$ only exists for $g_k \in \mathcal{G}_k$, that is, only localizable elements of $G$ give rise to a computational order parameter. 

\subsection{Logical observables}

To implement MBQC on symmetric states, we need to specify how information will be encoded and recovered from such states. To this end, we define the logical observables
\begin{equation} \label{eq:logi_obs}
    T(g) = 
    \begin{cases}
    V_L(g) \otimes U_{B}(g) \otimes V_R(g), & g \in H \\
    I_L \otimes U_{B}(g) \otimes V_R(g), & g \in G \setminus H
    \end{cases}
\end{equation}
where $I_L = \bigotimes_{i\in L} I_i$ is the identity on the left boundary and $H$ is a maximal subgroup of $G$ such that $T(H)$ is abelian. Overall, $T(G)$ forms a projective representation of $G$. However, because $T(H)$ commute they can be measured simultaneously and will be the basis states for our MBQC readout. They also define the initial state of the logical system, as shown in the following lemma.
\begin{lemma} \label{lemma:logi_init}
    For all states $\ket{\Phi}$ satisfying \eqref{eq:Ug}, 
    \begin{equation}
        \expval{T(g)}_\Phi = 
        \begin{cases}
            (-1)^{\chi(g)}, & g \in H \\
            0, & g \in G\setminus H
        \end{cases}
    \end{equation}
\end{lemma}
\begin{proof}
    For $g \in H$ we have $T(g)=U(g)$ and the result follows from \eqref{eq:Ug}. By assumption $T(H)$ is a maximal abelian subgroup of $T(G)$, thus for every $g \in G\setminus H$ there exists $h \in H$ such that $T(h)$ anticommutes with $T(g)$. Then $\bra{\Phi}T(g)\ket{\Phi} = \bra{\Phi}T(h)^\dagger T(g) T(h)\ket{\Phi} = -\bra{\Phi}T(g)\ket{\Phi} \implies \bra{\Phi}T(g)\ket{\Phi} = 0$. 
\end{proof}

\subsection{Computational power}

With $R_k(g_k)$, $\sigma_k(g_k)$, and $T(G)$ defined, we can state the following theorem regarding the computational power of symmetric states.

\begin{restatable}[Computational power]{thm}{power} \label{thm:power}
    Consider a state $\ket{\Phi}$ satisfying the following conditions:
    \begin{enumerate}[(a)]
        \item \label{item:finite spread} $\ket{\Phi}$ is related to a stabilizer state $\ket{\Phi_0}$ by a local, finite-depth unitary that commutes with $U(G)$ and has finite spread $\Delta$. 
        \item \label{item:good stabs} There exists a finite $d$ such that Lemma \ref{lemma:decorr} holds for any two regions of the lattice separated by at least $2\Delta+d$.
        \item \label{item:g symmetric} $\ket{\Phi}$ is symmetric under $U(G)$ as defined in \eqref{eq:Ug}.
        \item \label{item:comp order nonzero} The computational order parameters $\sigma_{k}(g_k)$ are non-zero for all $k$. 
    \end{enumerate}
    Then MBQC with $\ket{\Phi}$ as a resource state can implement the following:
    \begin{enumerate}[(i)]
        \item \label{thm_item:init} Initialization and readout of $\rank(H)$ logical qubits, where $H$ is a maximal subgroup of $G$ such that $T(H)$ is abelian.
        \item \label{thm_item:logi} Logical gates from the Lie group $\mathcal{L} = \exp(-i\mathcal{A})$, where $\mathcal{A}$ is the Lie algebra generated by $T(\mathcal{G})$ under $i\comm{\cdot}{\cdot}$ and linear combinations.
    \end{enumerate}
\end{restatable}

We prove Theorem \ref{thm:power} in Appendix \ref{sec:power proof}. The key takeaway is that every symmetric state such that the order parameters $\sigma_k(g_k)$ are non-zero has equivalent computational power, in the sense that they enable the same Lie group of logical gates for MBQC. Conditions \ref{item:finite spread} and \ref{item:good stabs} generalize the notion of a short-range entangled state to allow for states with topological order, and the quantity $d$ depends on the details of the stabilizer state $\ket{\Phi_0}$, for example the toric code stabilizers as shown in Fig. \ref{fig:bdaries}.

As background to the theorem, we sketch the process of MBQC on symmetric states below.

\paragraph{Initialization:} The initial state of the logical system is an eigenstate of $T(H)$ with eigenvalues given by Lemma \ref{lemma:logi_init}.

\paragraph{Readout} After measuring all sites on the left boundary and in the bulk, the measurement statistics for $T(g)$ can be reproduced by measuring $V_R(g)$ on the right boundary. $V_R(H)$ can be inferred simultaneously from the single-qubit measurement outcomes of $\hat{v}_i^R(H)$.

\paragraph{Logical rotations:} For every site $k$ in the bulk there is a symmetry-respecting measurement basis, namely the eigenbasis of the non-identity $\hat{u}_k(G)$. Recall that $\hat{u}_k(G)$ are Paulis and commute, so there must be a unique $\hat{\alpha}_k \in \{X_k, Y_k, Z_k\}$ whose eigenbasis is the symmetry-respecting basis. If $\hat{\beta}_k(g_k)$, $R_k(g_k)$, $\sigma_k(g_k)$ exist for some $g_k \in \mathcal{G}$, then measurement in the rotated eigenbasis of 
    \begin{equation}
        O_k(g_k, \theta_k) = e^{-i\frac{\theta_k}{2}\hat{\beta}_k(g_k)} \hat{\alpha}_k e^{i\frac{\theta_k}{2}\hat{\beta}_k(g_k)}
    \end{equation}
implements the logical CPTP map
\begin{equation}
    \mathcal{M}_k(g_k, \theta_k) = \frac{1+\sigma_k(g)}{2}\qty[\exp(-i\frac{\theta_k}{2}T(g_k))] + \frac{1-\sigma_k(g_k)}{2}\qty[\exp(i\frac{\theta_k}{2}T(g_k))].
\end{equation}
where square brackets denote superoperators. This map acts on the logical observables, that is, it takes
\begin{equation}
    T(g) \to T'(g) = \mathcal{M}_k(g_k,\theta_k)(T(g)), \qquad g_k\in\mathcal{G}_k, g \in G.
\end{equation}
When $\sigma_k(g_k) = 1$, $\mathcal{M}_k$ simplifies to a unitary rotation 
\begin{equation}
    M_{\log,k}(g_k, \theta_k) = \exp(-i\frac{\theta_k}{2}T(g_k)),
\end{equation}
and when $\sigma_k(g_k) = 0$ or $\theta_k=0$ then $\mathcal{M}_k$ is the identity. When $0 < \sigma_k(g_k) < 1$ we get a linear combination of $M_{\log,k}(g_k, \theta_k)$ and $M_{\log,k}(g_k, -\theta_k)$. To approximately recover a unitary gate we can split up the desired rotation by $\theta_k$ into $N$ consecutive rotations by $\theta_k/N$. This approximates $M_{\log,k}(g_k,\theta_k)$ up to an error $\epsilon$ under the diamond norm with the properties stated in Lemma 12 of Ref.~\cite{aharonov1998quantum}. The error is given by 
\begin{equation}
    \epsilon = \frac{\theta_k^2}{N}\frac{1-\sigma_k(g_k)^2}{\sigma_k(g_k)^2}.
\end{equation}
Thus, while every state with nonzero order parameters has equivalent computational power, they vary in terms of computational efficiency. The most efficient states for computation are those for which all the order parameters are equal to one. 

Finally, given the ability to implement the logical rotations $M_{\log,k}(g_k, \theta_k)$ for all $g_k \in \mathcal{G}_k$, we can also implement rotations about their commutators via
\begin{equation}
    \exp(-i\frac{d\theta}{2}T(g_k))\exp(-i\frac{d\theta}{2}T(g_{k'})) \approx \exp(-\frac{(d\theta)^2}{4}\comm{T(g_k)}{T(g_{k'})}).
\end{equation}
This gives us access to all the unitaries in $\mathcal{L} = \exp(-i\mathcal{A})$, where $\mathcal{A}$ is the Lie algebra generated by $T(\mathcal{G})$.



\section{Computational versus physical order in the toric code} 
\label{sec:tc example}

\subsection{Computational order}
\label{sec:tc comp a}

\subsubsection{Boundary conditions for MBQC}
\label{sec:tc comp}
\begin{figure}
    \centering
    \includegraphics[width=0.75\linewidth]{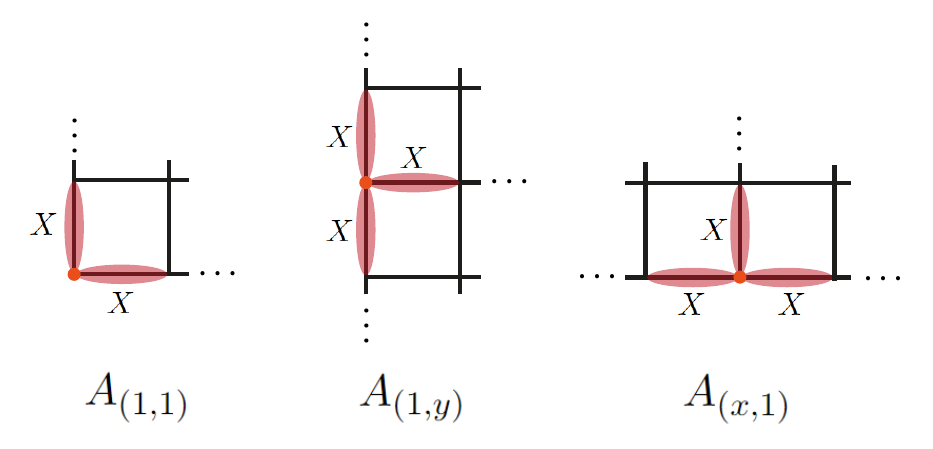}
    \caption{With open boundary conditions, we modify the vertex stabilizers by cutting them off at the edges.}
    \label{fig:cutoff}
\end{figure}

The MBQC scheme described in Section \ref{sec:comp order} requires a lattice with open boundary conditions, so we modify the toric code Hamiltonian Eq. \eqref{eq:TCH} to have smooth boundaries~\cite{bravyi1998quantumcodeslatticeboundary} on all sides. 
The plaquette stabilizers $B_{(x,y)}$ are unchanged, but we modify the vertex stabilizers $A_{(x,y)}$ by cutting off any tensor factors that fall outside the lattice (Fig. \ref{fig:cutoff}). 
We set the magnetic field strengths $h_X = h_Z = \alpha$, and we also remove the $X_{(x,y+\frac{1}{2})}$ terms at the left and right boundaries so that the system is invariant under a global linear representation $U(G)$ with projective representations at the boundaries, as required by Theorem \ref{thm:power}. 

Putting all these changes together yields the Hamiltonian 
\begin{equation} \label{eq:tch_mbqc}
    H_{\mathrm{MBQC}} = -\sum_{y=1}^{L_y}\sum_{x=1}^{L_x}  A_{(x,y)} - \sum_{y=1}^{L_y-1}\sum_{x=1}^{L_x-1} B_{(x,y)} -\alpha\sum_{y=1}^{L_y-1}\sum_{x=2}^{L_x-1} X_{(x,y+\frac{1}{2})} - \alpha \sum_{y=1}^{L_y}\sum_{x=1}^{L_x-1} Z_{(x+\frac{1}{2},y)}.
\end{equation}
Note that $H_{\mathrm{MBQC}}$ has a unique ground state $\ket{\psi_\alpha}$ for all finite values of $\alpha$. These ground states will be our resource states for MBQC.

\subsubsection{Symmetries and order parameters}
\begin{figure}
    \centering
    \includegraphics[width=0.75\linewidth]{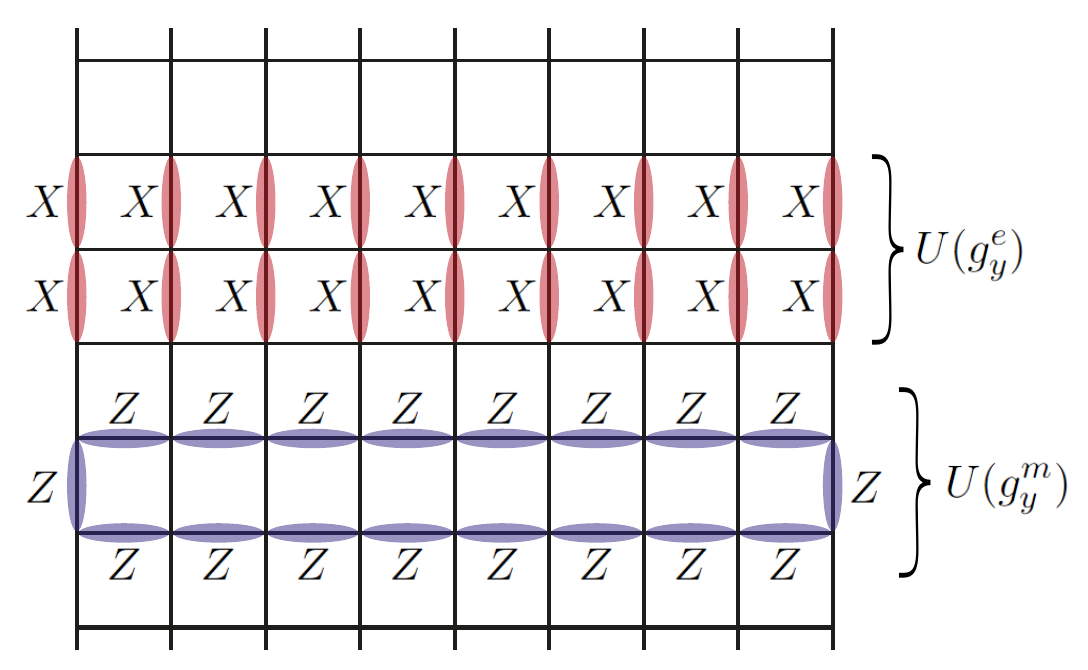}
    \caption{Action of the subsystem symmetries $U(G)$ for the toric code with smooth open boundaries.}
    \label{fig:TC UG}
\end{figure}

\paragraph{Symmetries} $H_{\mathrm{MBQC}}$ and its ground states $\ket{\psi_\alpha}$ are invariant under the same subsystem symmetries defined in \eqref{eq:subsystem_symmetries}, with modifications near the smooth boundary to ensure that they commute with the boundary terms. Specifically, we have
\begin{equation} \label{eq:tc_ug}
    U(g_y^e) = \prod_{x=1}^{L_x} A_{(x,y)}, \qquad U(g_y^m) = \prod_{x=1}^{L_x-1} B_{(x,y)}
\end{equation}
The smooth boundary conditions create one more row of vertices than plaquettes, yielding the symmetry group $G = \langle g_y^e, g_y^m \rangle \cong \mathbb{Z}_2^{2L_y-1}$. 

At the right boundary, we can find the projective representations by looking at how $U(g_y^e)$ and $U(g_y^m)$ terminate (Fig. \ref{fig:TC UG}). This yields
\begin{equation} \label{eq:tc_proj_reps}
    V_R(g_y^e) = 
    \begin{cases}
        X_{(L_x, \frac32)}, & y=1, \\
        X_{(L_x,y-\frac12)}X_{(L_x, y+\frac12)}, & 1 < y < L_y \\ 
        X_{(L_x, L_y-\frac12)}, & y=L_y.
    \end{cases} 
    ,\qquad V_R(g_y^m) = Z_{(L_x, y+\frac12)}. 
\end{equation}
We can see that this is indeed a projective representation because $V_R(g_y^e)$ anticommutes with $V_R(g_y^m)$ and $V_R(g_{y-1}^m)$, that is $\kappa(g_y^e, g_y^m) = \kappa(g_y^e, g_{y-1}^m) = 1$. The projective representations on the left boundary are the same.

\paragraph{Computational order parameters} Now let us localize one of the linear representations $U(g_y^e)$ to create a computational order parameter. First, we restrict $U(g_y^e)$ to the region right of site $(k+\frac12,y)$ with $1 < k < L_x$, yielding 
\begin{equation}
    U_{>k+\frac12}(g_y^e) = X_{(k+\frac12,y)}\prod_{x=k+1}^{L_x} A_{(x,y)}.
\end{equation}
 $U_{>k+\frac12}(g_y^e)$ anticommutes with $U(g_y^m)$ and $U(g_{y-1}^m)$ due to the projective representations at the right boundary, but \eqref{eq:localizability condition} says that we may be able to ``fix up'' this anticommutation by adding an operator $\hat{\beta}_{k+\frac12, y}(g_y^e)$ at site $(k+\frac12, y)$. On this site, $U(g_y^m)$ and $U(g_{y-1}^m)$ act by $Z$, so we can choose $\hat{\beta}_{k+\frac12, y}(g_i^e) = X_{k+\frac12, y}$. Then the product
 \begin{equation} \label{eq:tc cop e}
     R_{k+\frac12}(g_y^e) = X_{(k+\frac12, y)} \otimes U_{>k+\frac12}(g_y^e) = \prod_{x=k+1}^{L_x} A_{(x,y)} 
 \end{equation}
 commutes with all $U(G)$. This gives us an ``electric'' computational order parameter $\sigma_{k+\frac12}(g_y^e) = \expval{R_{k+\frac12}(g_y^e)}$. Note that there are no other $U(g)$ which can be localized at site $(k+\frac12, y)$, thus $\mathcal{G}_{k+\frac12,y} = \{g_y^e\}$.
 
 Similarly, we can localize $U(g_y^m)$ to any site $(k, y+\frac12), 2\leq k \leq L_x-1$ and obtain
 \begin{equation} \label{eq:tc cop m}
     R_{k}(g_y^m) = \prod_{x=k}^{L_x-1} B_{(x,y)},
 \end{equation}
 from which we get a ``magnetic'' computational order parameter $\sigma_{k}(g_y^m) = \expval{R_{k}(g_y^m)}$. Once again, $U(g_y^m)$ is the only symmetry that can be localized to site $(k,y+\frac12)$, so we have $\mathcal{G}_{k, y+\frac12} = \{g_y^m\}$. Taking the union $\bigcup_{x,y} \mathcal{G}_{x,y}$ we have $\mathcal{G} = \{g_y^e, g_y^m:y=1,...,L_y\}$. The operators $R_{k+\frac12}(g_y^e)$ and $R_{k}(g_y^m)$ are shown in Fig. \ref{fig:TC Rk TG}.

 \begin{figure}
     \centering
     \includegraphics[width=0.99\linewidth]{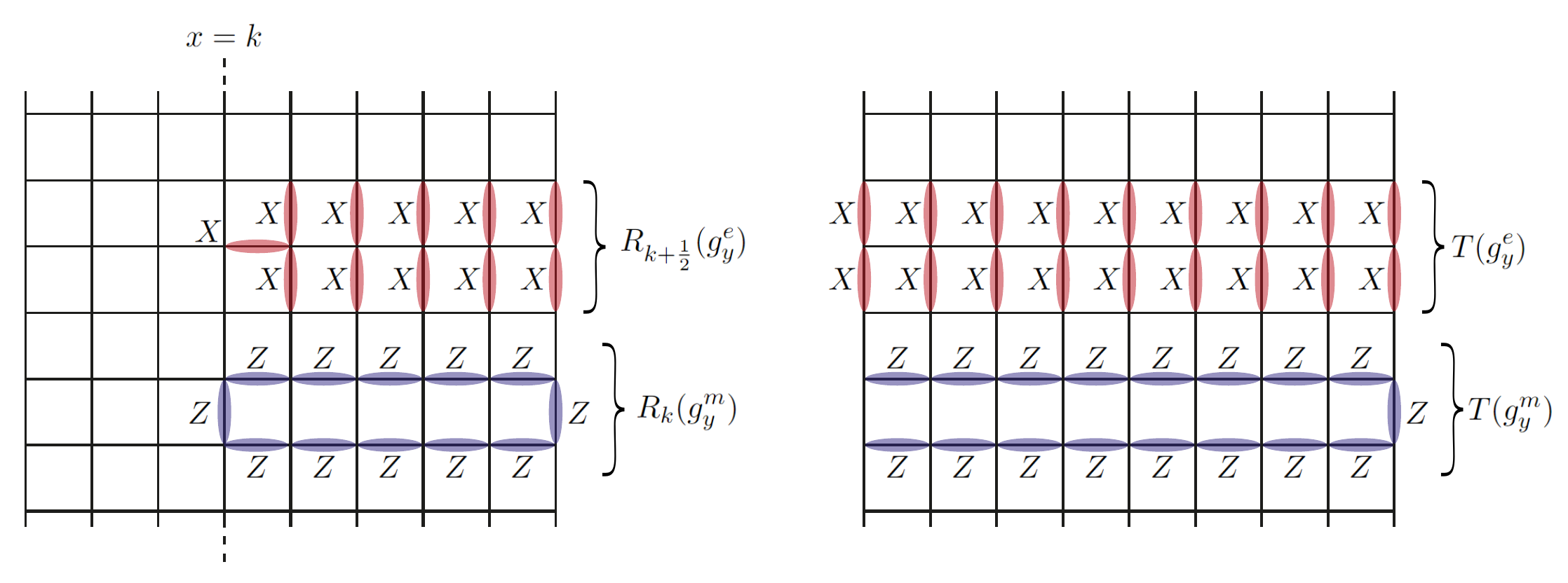}
     \caption{\textbf{Left}: operators $R_{k+\frac12}(g_y^e)$ and $R_{k}(g_y^m)$ that give rise to computational order parameters in the toric code. \textbf{Right}: logical observables $T(g_y^e)$ and $T(g_y^m)$.}
     \label{fig:TC Rk TG}
 \end{figure}

\paragraph{Logical observables}

Recall that we have one more row of vertices than we do of plaquettes, therefore a maximal subgroup $H\subset G$ such that $V_R(H)$ is abelian is given by $H = \expval{g_y^e:i=1,...,m+1}$. With respect to $H$, we can define the logical observables 
\begin{equation} \label{eq:tc t ops}
    T(g_y^e) = U(g_y^e) = \prod_{x=1}^{L_x} A_{(x,y)}, \qquad T(g_y^m) = Z_{(1, y+\frac12)} \prod_{x=1}^{L_x-1} B_{(x,y)}.
\end{equation}
The observables $T(g_y^e)$ all commute, so we use them as $L_y$ logical qubits for initialization and readout. 

\paragraph{Logical gates} The logical rotation associated with the electric COPs is $\exp(i\frac{\theta}{2}T(g_y^e))$, and the logical rotation for the magnetic COPs is $\exp(i\frac{\theta}{2}T(g_y^m))$. These are all the logical rotations that exist, because we have establised $\mathcal{G} = \{g_y^e, g_y^m:y=1,...,L_y\}$. The anticommutation relations of $T(g_y^e)$ and $T(g_y^m)$ are completely determined by the projective representations $V_R(g_y^e)$ and $V_R(g_y^m)$ at the right boundary, so that the Lie algebra $\mathcal{A}$ generated by $T(\mathcal{G})$ is isomorphic to the Lie algebra generated by $V_R(\mathcal{G})$. If we think of the right boundary as a 1D system of $L_y-1$ qubits then \eqref{eq:tc_proj_reps} tells us that $V_R(\mathcal{G})$ acts as
\begin{equation} \label{eq:tc gens}
    \{X_1, X_{1}X_{2}, ..., X_{L_y-2}X_{L_y-1}, X_{L_y-1}, Z_1, ..., Z_{L_y-1}\}.
\end{equation}
These operators can be mapped via a Jordan-Wigner transformation to the creation operators for a system of $L_y$ free Majorana fermions, which generate the special orthogonal Lie algebra $\so(2L_y)$. Therefore, we have $\mathcal{A}\cong\so(2L_y)$ and the Lie group of gates available for MBQC is $\mathcal{L} = \SO(2L_y)$. 
\paragraph{Computational power}

Having derived the global symmetries and COPs, we can check that the conditions of Theorem \ref{thm:power} are satisfied for the ground states $\ket{\psi_\alpha}$.
\begin{enumerate}[(a)]
    \item $\ket{\psi_\alpha}$ are gapped ground states of the Hamiltonian \eqref{eq:tch_mbqc}. In the thermodynamic limit any given $\ket{\psi_\alpha}$ is in the same phase as the toric code state $\ket{\psi_{TC}}$ or a product state, both of which are stabilizer states. It is known that two states are in the same topological phase if and only if they are related by a local, finite-depth unitary~\cite{Bravyi_2006}, thus for any $\alpha$ there exists a local, finite-depth unitary connecting it to a stabilizer state. Such a unitary necessarily has finite spread. 
    \item As discussed in Section \ref{sec:decorr}, we have $d=1.5$ such that regions with any combination of rough or smooth boundaries will satisfy Lemma \ref{lemma:decorr} if they are separated by at least $2\Delta+d$.
    \item $\ket{\psi_\alpha}$ are symmetric under the linear representation $U(G)$ as defined in \eqref{eq:tc_ug}.
\end{enumerate}

With conditions (a), (b) and (c) fulfilled, Theorem \ref{thm:power} tells us that any $\ket{\psi_\alpha}$ such that the computational order parameters are nonzero can be used as a resource state for MBQC with $L_y$ logical qubits and gates drawn from $\SO(2L_y)$. In the next section, we show that every computational order parameter is dual to a string order parameter as defined in \eqref{eq:sop def vir}, and therefore that the ground-state topological phase of $H_{\mathrm{MBQC}}$ including the toric code corresponds to a ``computational phase of matter''~\cite{Raussendorf2019}.

\subsection{Topological order}
\label{sec:tc topo}

\subsubsection{Derivation of the phase diagram} 
Having understood the ground-state computational power of the Hamiltonian Eq. \eqref{eq:TCH}, we now investigate its ground-state topological order with a view to showing how these two properties are related. 
To understand the ground-state physics of the Hamiltonian Eq. \eqref{eq:TCH}, it is simplest to borrow the language of lattice gauge theory~\cite{kogut1979introduction,kitaev2003fault} and split $H$ into ``electric'' and ``magnetic'' contributions,
\begin{equation}
\label{eq:EandM}
H = H^e + H^m, \quad H^e = - \sum_{(x,y) \in \Lambda} A_{(x,y)} - h_Z \sum_{(x,y)\in\Lambda} Z_{(x+\frac12,y)}, \quad H^m = - \sum_{(x,y)\in\Lambda} B_{(x,y)} - h_X \sum_{(x,y)\in\Lambda} X_{(x,y+\frac12)}.
\end{equation}
A distinctive feature of our anisotropic perturbing magnetic field Eq. \eqref{eq:dH} compared to the standard~\cite{TCPD} isotropic perturbing magnetic field of the toric code is that the electric and magnetic terms in Eq. \eqref{eq:EandM} commute with one another, $[H^e,H^m]=0$. In fact, these Hamiltonians are completely decoupled from one another. This is intuitively clear in the language of anyonic excitations of $H_0$~\cite{kitaev2003fault,kitaev2009topologicalphasesquantumcomputation}, according to which the operators $A_{(x,y)}$ ($B_{(x,y)}$) count the number of electric charges (magnetic vortices) at their centre, while the operators $Z_{(x+\frac12,y)}$ ($X_{(x,y+\frac12)}$) create perfectly horizontal pairs of electric charges (magnetic vortices) on either side of the bonds that they act on. Thus, unlike the standard perturbed toric code, our Hamiltonian $H$ is equivalent to $L_y$ mutually decoupled pairs of Ising models acting on electric and magnetic anyons respectively. See bottom panel of Fig. \ref{fig:anyons}. Let us now make this statement mathematically precise.

We can formalize the qualitative mapping described above by taking $L_x \to \infty$ and mapping electric or magnetic anyons to Ising domain walls defined in terms of new Pauli spin variables $\sigma$ (on horizontal bonds of $\Lambda$) and $\tau$ (on sites of $\Lambda$) respectively, as
\begin{equation}
A_{(x,y)} \mapsto \sigma^z_{(x-\frac{1}{2},y)} \sigma^z_{(x+\frac{1}{2},y)}, \quad Z_{(x+\frac{1}{2},y)} \mapsto \sigma^x_{(x+\frac{1}{2},y)},
\end{equation}
in the electric sector and
\begin{equation}
B_{(x,y)} \mapsto \tau^z_{(x,y)} \tau^z_{(x+1,y)}, \quad X_{(x,y+\frac{1}{2})} \mapsto \tau^x_{(x,y)}
\end{equation}
in the magnetic sector~\footnote{Equivalently, one can define the $\sigma$ and $\tau$ operators microscopically as $
\sigma^z_{(x+\frac{1}{2},y)} = \prod_{x' \leq x} A_{(x,y)}$ and $\sigma^x_{(x+\frac{1}{2},y)} = Z_{(x+\frac{1}{2},y)}$ in the electric sector and $
\tau^z_{(x,y)} = \prod_{x' < x} B_{(x,y)}$ and $\tau^x_{(x,y)} = X_{(x,y+\frac{1}{2})}$ in the magnetic sector.}. In terms of these effective spin variables, we find that
\begin{equation}
\label{eq:Isingstack}
H^e = \sum_{y=1}^{L_y} H^e_y, \quad H^m = \sum_{y=1}^{L_y} H^m_y,
\end{equation}
where $[H^e_y,H^e_{y'}] =[H^e_y,H^m_{y'}]=[H^m_y,H^m_{y'}]=0$ for all $y,y'=1,2,\ldots, L$, with
\begin{equation}
\label{eq:directE}
H^e_y = -\sum_{x\in \mathbb{Z}} \sigma^z_{(x-\frac{1}{2},y)} \sigma^z_{(x+\frac{1}{2},y)} - h_Z \sum_{x \in \mathbb{Z}} \sigma^x_{(x+\frac{1}{2},y)}
\end{equation}
and
\begin{equation}
\label{eq:directB}
H^m_y = -\sum_{x \in \mathbb{Z}} \tau^z_{(x,y)} \tau^z_{(x+1,y)} - h_X \sum_{x\in \mathbb{Z}}\tau^x_{(x,y)}.
\end{equation}
We have thus mapped the Hamiltonian $H$ to a stack of decoupled transverse-field Ising chains. Similar Hamiltonians and mapping to decoupled Ising chains were previously obtained for the Wen plaquette model~\cite{WenPlaquette} and for the bond-anisotropic Kitaev model~\cite{wachtel2019confinement}. It follows that the confinement transitions out of the topological ground-state phase of $H$ as $h_X$ and $h_Z$ are increased from zero are equivalent to the ferromagnetic-to-paramagnetic ground-state phase transition of the exactly solvable~\cite{pfeuty1970one} one-dimensional transverse field Ising model. We deduce that for all values of $L_y$, electric charges condense (and magnetic vortices are confined) when $h_Z > 1$ and magnetic vortices condense (and electric charges are confined) when $h_X > 1$, leading to the simple phase diagram shown in the right panel of Fig. \ref{fig:anyons} as $L_y \to \infty$. Note that in our discussion of measurement-based quantum computing above, we have set $h_X = h_Z = \alpha$ so that the confinement transition out of the topological phase occurs at the (trivially) bicritical point $\alpha=1$, see Fig. \ref{fig:anyons}.

\subsubsection{Topological order parameters}
We have shown that the Hamiltonian $H$ exhibits a topological ground-state phase in the thermodynamic limit for $|h_X|, \, |h_Z| < 1$. Let us now briefly recall how topological phases are diagnosed in practice, for example, for the usual magnetic-field perturbation of the toric code~\cite{TCPD,xu2024criticalbehaviorfredenhagenmarcustring}. The key idea is ``string order''~\cite{bricmont1983order,fredenhagen1983charged,fredenhagen1986confinement,marcu1987new} that motivates the construction of Fredenhagen-Marcu order parameters~\cite{fredenhagen1983charged,fredenhagen1986confinement,marcu1987new} in lattice gauge theory. Such order parameters have proved theoretically~\cite{huse1991sponge,Gregor_2011}, experimentally~\cite{SLexpt} and numerically~\cite{xu2024criticalbehaviorfredenhagenmarcustring} valuable for diagnosing topological order in condensed matter systems.

For concreteness, let $L_x = L_y = L$ and let $|\psi_{\mathrm{TC}}\rangle$ denote the ground state of the unperturbed toric code Hamiltonian $H_0$. For any (closed or open) string $\Gamma^e$ on the bonds of the direct lattice or $\Gamma^m$ on the bonds of the dual lattice, define string operators or Wilson lines
\begin{equation} \label{eq:sop def vir}
W[\Gamma^e]  = \prod_{b \in \Gamma^e} Z_b, \quad W[\Gamma^m]  = \prod_{b \in \Gamma^m} X_b.
\end{equation}
When the strings in question are open, the corresponding string operators act on $|\psi_{\mathrm{TC}}\rangle$ to create pairs of anyons at their endpoints, so that the resulting state is in a different $\mathbb{Z}_2$-charge sector from the ground state and $\langle \psi_{\mathrm{TC}}| W[\Gamma^{e/m}] | \psi_{\mathrm{TC}} \rangle = 0$. For open strings with length $\Omega(L)$, one might similarly expect that $\langle \psi | W[\Gamma^{e/m}] | \psi \rangle \to 0$ in the thermodynamic limit whenever the relevant anyons are deconfined i.e. $|\psi\rangle$ is in a ``topological'' phase, while $\langle \psi | W[\Gamma^{e/m}] | \psi \rangle > 0$ in the thermodynamic limit whenever the relevant anyons are confined, i.e. $|\psi\rangle$ is in a ``trivial'' phase. This motivates the following definition: let $\Gamma^{e/m}_o$ be open strings of length $\Omega(L)$. Then we define ``string order parameters''
\begin{equation}
O[\Gamma^{e/m}_o] = \lim_{L\to\infty}\sqrt{|\langle \psi | W[\Gamma^{e/m}_o] | \psi \rangle|}.
\end{equation}
However, for the usual magnetic field perturbation of the toric code, the string order parameter has the fatal flaw that it is susceptible to ``perimeter law'' scaling at generic points of the phase diagram and therefore decays too fast to be distinguishable from zero in practice, even in the trivial phase. This is usually remedied by the prescription due to Fredenhagen and Marcu~\cite{fredenhagen1983charged,fredenhagen1986confinement,marcu1987new}, according to which additional closed loops $\Gamma^{e/m}_c$ with length $\Gamma^{e/m}_c = 2 |\Gamma^{e/m}_o|$ are introduced, yielding the definition
\begin{equation}
\tilde{O}[\Gamma^{e/m}_o] = \lim_{L\to\infty}\sqrt{\frac{|\langle \psi | W[\Gamma^{e/m}_o] | \psi \rangle|}{\sqrt{|\langle \psi | W[\Gamma^{e/m}_c] | \psi \rangle|}}}.
\end{equation}
This rescaling of the string order parameter cancels out the decay due to the perimeter law, yielding a practically useful tool for diagnosing topological order~\cite{Gregor_2011,SLexpt,xu2024criticalbehaviorfredenhagenmarcustring}.

Note that the phase diagram in Fig. \ref{fig:anyons} differs from this conventional understanding in the important respect that its perimeter-law scaling is highly anistropic: only vertical bonds of the direct lattice contribute to perimeter-law scaling of $O(\Gamma^e_o)$, while only vertical bonds of the dual lattice contribute to perimeter-law scaling of $O(\Gamma^m_o)$. What this means in practice is that the string order parameters $O(\Gamma^{e/m}_o)$ suffice to diagnose the various ground-state phases of $H$ without Fredenhagen-Marcu-type rescaling, provided the long open strings $\Gamma^{e/m}_o$ are taken to be perfectly horizontal. In the remainder of this paper, ``string operators'' and ``string order parameters'' will always refer to this special case of perfectly horizontal strings. 

\subsubsection{Kramers-Wannier duality}
We showed above that the Hamiltonian $H$ could be solved exactly by mapping the bare toric code operators to new Pauli operators $\sigma$ and $\tau$ via the map
\begin{equation}
\label{eq:directmapping}
\begin{aligned}
M\big(A_{(x,y)}\big) &= \sigma^z_{(x-\frac{1}{2},y)} \sigma^z_{(x+\frac{1}{2},y)}, \\ 
M\big(Z_{(x+\frac{1}{2},y)}\big) &= \sigma^x_{(x+\frac{1}{2},y)},\\ 
M\big(B_{(x,y)}\big) &= \tau^z_{(x,y)} \tau^z_{(x+1,y)}, \\
M\big(X_{(x,y+\frac{1}{2})}\big) &= \tau^x_{(x,y)}.
\end{aligned}
\end{equation}
However, the image of this map famously exhibits Kramers-Wannier duality~\cite{kramers1941statistics}, implying two distinct ways of mapping to Pauli operators for each row of $\Lambda$ and for each sector (electric or magnetic) within a given row, yielding $2^{2L_y}$ distinct ways of mapping $H$ to a stack of transverse-field Ising models. If we demand that the same choice is made in every row and between electric and magnetic sectors, this vast freedom reduces to two possible choices, namely the mapping Eq. \eqref{eq:directmapping} and its Kramers-Wannier dual map $\tilde{M}$, involving Pauli spin variables $\tilde{\sigma}$ (on sites of $\Lambda$) and $\tilde{\tau}$ (on horizontal bonds of $\Lambda$) such that
\begin{equation}
\label{eq:dualmapping}
\begin{aligned}
\tilde{M}\big(A_{(x,y)}\big) &= \tilde{\sigma}^x_{(x,y)}, \\ 
\tilde{M}\big(Z_{(x+\frac{1}{2},y)}\big) &= \tilde{\sigma}^z_{(x,y)}\tilde{\sigma}^z_{(x+1,y)},\\ 
\tilde{M}\big(B_{(x,y)}\big) &= \tilde{\tau}^x_{(x+\frac{1}{2},y)}, \\
\tilde{M}\big(X_{(x,y+\frac{1}{2})}\big) &= \tilde{\tau}^z_{(x-\frac{1}{2},y)}\tilde{\tau}^z_{(x+\frac{1}{2},y)}.
\end{aligned}
\end{equation}
Under this dual map, we find that the decomposition into a stack of independent transverse-field Ising models Eq. \eqref{eq:Isingstack} continues to hold up to an exchange of coupling constants, i.e.
\begin{equation}
\label{eq:hedual}
H^{e}_y = - h_Z \sum_{x\in \mathbb{Z}} \tilde{\sigma}^z_{(x,y)}\tilde{\sigma}^z_{(x+1,y)} - \sum_{x \in \mathbb{Z}} \tilde{\sigma}^x_{(x,y)} 
\end{equation}
and 
\begin{equation}
\label{eq:hmdual}
H^{m}_y = - h_X \sum_{x\in \mathbb{Z}} \tilde{\tau}^z_{(x-\frac{1}{2},y)}\tilde{\tau}^z_{(x+\frac{1}{2},y)} - \sum_{x \in \mathbb{Z}} \tilde{\tau}^x_{(x+\frac{1}{2},y)}.
\end{equation}
Comparison with Eqs. \eqref{eq:directE} and \eqref{eq:directB} reveals that if there is a unique phase transition along the ``self-dual line'' $h_Z = h_X = \alpha > 0$ in Fig. \ref{fig:anyons}, it must occur at $\alpha=1$, which provides a non-perturbative way of understanding the location of this phase transition. 

In terms of the Kramers-Wannier duality mapping $D$, which can be defined explicitly as
\begin{equation}
\begin{aligned}
D\big(\sigma^z_{(x-\frac{1}{2},y)} \sigma^z_{(x+\frac{1}{2},y)}\big) &= \tilde{\sigma}^x_{(x,y)}, \\ 
D\big(\sigma^x_{(x+\frac{1}{2},y)}\big) &= \tilde{\sigma}^z_{(x,y)}\tilde{\sigma}^z_{(x+1,y)},\\ 
D\big(\tau^z_{(x,y)} \tau^z_{(x+1,y)}\big) &= \tilde{\tau}^x_{(x+\frac{1}{2},y)}, \\
D\big(\tau^x_{(x,y)}\big) &= \tilde{\tau}^z_{(x-\frac{1}{2},y)}\tilde{\tau}^z_{(x+\frac{1}{2},y)},
\end{aligned}
\end{equation}
we have
\begin{equation}
\tilde{M} = D \circ M.
\end{equation}



\subsubsection{Subsystem symmetry and SSET order} \label{sec:sset}

As discussed in the introduction, the Hamiltonian $H$ can also be understood as a model with subsystem symmetry. Namely, because of the anisotropy of the magnetic field, $H$ commutes with all Wilson loops $\Gamma^e_c$ ($\Gamma^m_c$) where $c$ is a perfectly horizontal closed loop on the direct (dual) lattice. There are two such operators for each row $y=1,2,\ldots,L_y$. Thus the Hamiltonian $H$ enjoys a much higher degree of symmetry than more generic~\cite{TCPD} magnetic-field perturbations of the toric code.

However, we do not consider all of these Wilson loops to be symmetries. This is because a single Wilson loop, when truncated, creates anyons at its endpoints. These anyons cannot be removed by acting only near the endpoints of the Wilson loop. Therefore, a single Wilson loop cannot be ``localized'' or used to define a string order parameter. This is due to the fact that the 1-form symmetry generated by the Wilson loops is spontaneously broken. A pair of adjacent Wilson loops, however, is preserved since it creates a pair of anyons at its endpoints, and this pair can be annihilated by acting with a local operator. Therefore, we define the subsystem symmetries of $H$ to be generated by all pairs of adjacent, horizontal Wilson loops. This matches the structure of the symmetries derived in Sec.~\ref{sec:tc comp a}. Furthermore, under the duality mapping in Eq.~\ref{eq:dualmapping}, such symmetries correspond to the $\mathbb{Z}_2$ symmetries of the individual Ising chains, which provides a second intuitive picture of the origin of the subsystem symmetries (alongside the conservation law along rows discussed in Sec.~\ref{sec:summary}).



\begin{figure}
    \centering
    \includegraphics[width=0.6\linewidth]{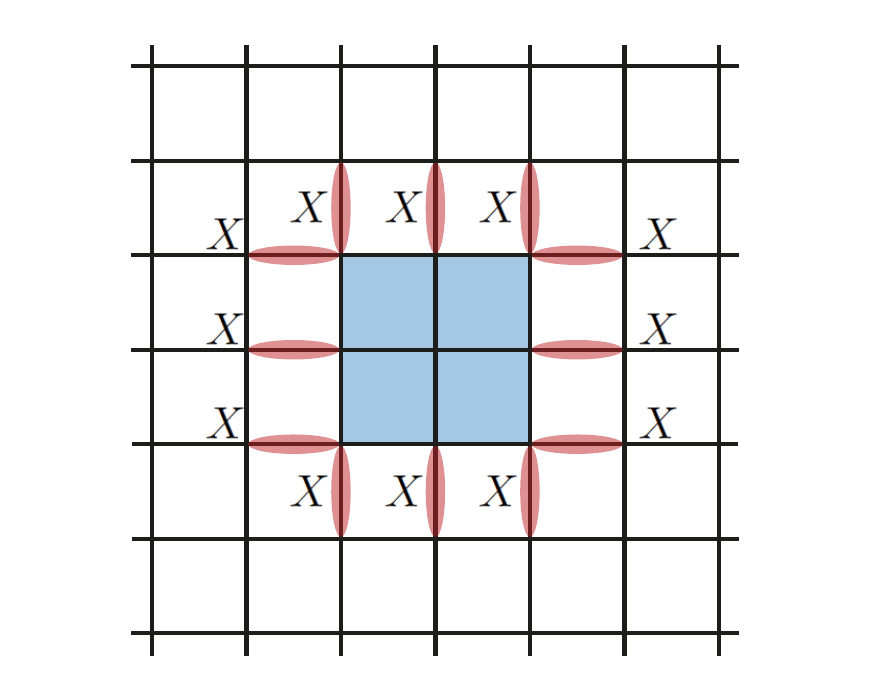}
    \caption{Truncating the global relation in the $e$-sector to a finite region $\mathcal{R}$ (shaded region) is equal to a Wilson loop around the boundary of that region. The truncation for the $m$-sector is similar.}
    \label{fig:relation}
\end{figure}

We now demonstrate that the toric code belongs to a non-trivial SSET phase protected by these symmetries.
For now, we return to the thermodynamic limit and define the symmetries more precisely as
\begin{equation} \label{eq:subsystem_symmetries_inf}
    S(g_y^e) = \prod_{x\in\mathbb{Z}} A_{(x,y)},\qquad S(g_y^m) = \prod_{x\in\mathbb{Z}} B_{(x,y)},
\end{equation}
which are equal to products of $X$ and $Z$, respectively, on neighbouring rows, as in Fig.~\ref{fig:subsystem_symms}.
These symmetries generate an infinite group. However, the symmetries are not all independent, as they satisfy a pair of \textit{global relations},
\begin{equation} \label{eq:global_relation}
    \prod_{y\in \mathbb{Z}} S(g_y^e) = \prod_{x,y\in\mathbb{Z}\times\mathbb{Z}} A_{(x,y)} = 1,\qquad \prod_{y\in \mathbb{Z}} S(g_y^m) = \prod_{x,y\in\mathbb{Z}\times\mathbb{Z}} B_{(x,y)} = 1
\end{equation} 
According to the procedure described in Ref.~\cite{Stephen2022}, identifying global relations between symmetry generators is the first step in determining which SSET phase a model belongs to. The next step is to truncate these global relations to finite regions. We first define the truncation of the symmetry generators to finite intervals $\mathcal{I}_x$,
\begin{equation}
    S_{\mathcal{I}_x}(g_y^e) = \prod_{x\in \mathcal{I}_x} A_{(x,y)},\qquad S_{\mathcal{I}_x}(g_y^m) = \prod_{x\in \mathcal{I}_x} B_{(x,y)},
\end{equation}
These truncations are defined such that they act the same as the untruncated generators inside the region $\mathcal{I}_x$ while still commuting with the Hamiltonian. They are in fact identical to the string-order parameters pictured in Fig.~\ref{fig:subsystem_symms}. Now, we truncate the global relation by restricting the product in Eq.~\ref{eq:global_relation} to a finite interval $\mathcal{I}_y$,
\begin{equation}
    \prod_{y\in \mathcal{I}_y} S_{\mathcal{I}_x}(g_y^e)  = \prod_{x,y\in \mathcal{I}_x\times \mathcal{I}_y} A_{(x,y)} \equiv W[\Gamma^m_{\partial \mathcal{R}}],\qquad  \prod_{y\in \mathcal{I}_y} S_{\mathcal{I}_x}(g_y^m)  = \prod_{x,y\in \mathcal{I}_x\times \mathcal{I}_y} B_{(x,y)} \equiv W[\Gamma^e_{\partial \mathcal{R}}]
\end{equation}
where $\mathcal{R}=\mathcal{I}_x\times \mathcal{I}_y$ is a finite rectangular patch on the lattice and $\partial \mathcal{R}$ is its boundary. We see that restricting the global relation to a finite patch $\mathcal{R}$ is equivalent to the action of a Wilson loop around the boundary of $\mathcal{R}$, as in Fig.~\ref{fig:relation}. More specifically, we have the following correspondence,
\begin{equation} \label{eq:invariants}
    \phi\left(  \prod_{y\in \mathbb{Z}} S(g_y^e)\right) = m, \qquad \phi\left(  \prod_{y\in \mathbb{Z}} S(g_y^m)\right) = e
\end{equation}
where the function $\phi$ takes in a global relation and returns the label of the superselection sector (i.e. anyon type) corresponding to the Wilson loop that appears at the boundary of the truncated relation. It is shown in Ref.~\cite{Stephen2022} that the function $\phi$ is invariant under symmetry-preserving perturbations and that it completely characterizes the pattern of subsystem symmetry fractionalization, and therefore the SSET phase of the model. The fact that $\phi$ is a non-trivial map shows that the subsystem symmetries have a non-trivial interplay with the anyons. In the present model, this is not so surprising since the symmetries themselves are products of Wilson loops. In more general cases, this interplay is not obvious but can nonetheless be revealed by the above truncation procedure \cite{Stephen2022}.

The non-trivial subsystem symmetry fractionalization in this model is related to the restricted mobility of its anyons. As we have seen in Fig.~\ref{fig:anyons}, the symmetry restrictions only allow anyons to be created in pairs belonging to the same row. Equivalently, symmetry-preserving operations can only move a single anyon horizontally. This is in fact a necessary consequence of subsystem symmetry fractionalization, as shown in Ref.~\cite{Stephen2022}.

Finally, we observe that the above calculations were performed entirely in a finite region of the bulk of the lattice. Therefore, even if we consider a finite lattice with arbitrary boundary conditions, the above procedure can still be performed in the bulk to recover the invariants in Eq.~\ref{eq:invariants}. This is even true in the case that the global relations are violated near the boundary of the lattice. This is akin to the observation that conventional SPT phases are still well-defined even when the bulk symmetry needs to be modified near the boundary to commute with the Hamiltonian. 

%
%
%
%
%

\subsection{Relating computational order to string order}
\label{sec:tc relations}
We would now like to relate the notion of computational order, as discussed above, to ground-state topological order of the Hamiltonian $H_{\mathrm{MBQC}}$. This immediately leads us to the following tension: MBQC pertains to finite systems by construction, while phases of matter are only sharply defined in the thermodynamic limit. \\

In order to make a meaningful comparison between the two concepts, we must therefore take a suitable limit of MBQC. The simplest procedure is to return to the thermodynamic limit of the toric code Hamiltonian $H$ in Eq. \eqref{eq:TCH} and consider the analogue of the ``loop operators'' Eqs. \eqref{eq:tc cop e} and \eqref{eq:tc cop m}, which in this setting correspond to horizontal, contractible Wilson loops,
\begin{equation}
R^e_{y} = \prod_{x=0}^{\ell} A_{(x,y)}, \quad R^m_y = \prod_{x=0}^{\ell-1} B{(x,y)}.
\end{equation}
We would like to compare these objects to purely horizontal string operators
\begin{equation}
\label{eq:horizontal strings}
W[\Gamma^e_o] = \prod_{x=0}^{\ell-1}  Z_{(x+\frac12,y)}, \quad W[\Gamma^m_o] = \prod_{x=0}^{\ell}  X_{(x,y+\frac12)}.
\end{equation}
Under the mapping $M$, we find that the loop operators detect ferromagnetic correlations in the Ising model,
\begin{equation}
\label{eq:MtoComp}
M(R^e_y) = \sigma^z_{(-\frac12,y)}\sigma^z_{(l+\frac12,y)}, \quad  M(R^m_y) = \tau^z_{(0,y)} \tau^z_{(\ell,y)},
\end{equation}
while the string operators that diagnose the topological order map to Ising string operators of length order $\ell$,
\begin{equation}
M(W[\Gamma^e_o]) = \prod_{x'=0}^{\ell-1} \sigma^x_{(x'+\frac12,y)}, \quad M(W[\Gamma^m_o]) = \prod_{x'=0}^{\ell} \tau^x_{(x',y)}.
\end{equation}
This might na{\"i}vely suggest that the computational and topological order parameters are probing distinct kinds of order in this model. However, under Kramers-Wannier duality these two kinds of order are interchanged. Specifically, under the mapping $\tilde{M}$ we find that loop operators now map to Ising-string-type operators,
\begin{equation}
\tilde{M}(R^e_y) = \prod_{x'=0}^{\ell} \tilde{\sigma}_{(x',y)}^x, \quad  \tilde{M}(R^m_y) = \prod_{x'=0}^{\ell-1} \tilde{\tau}^x_{(x'+\frac12,y)}
\end{equation}
while string operators now detect ferromagnetic correlations,
\begin{equation}
\label{eq:MtildetoString}
\tilde{M}(W[\Gamma^e_o]) = \tilde{\sigma}^z_{(0,y)} \tilde{\sigma}^z_{(\ell,y)}, \quad \tilde{M}(W[\Gamma^m_o]) = \tilde{\tau}^z_{(-\frac12,y)} \tilde{\tau}^z_{(\ell+\frac12,y)}.
\end{equation}
Let us now take stock of what we have shown. Taking $\ell \to \infty$, Eq. \eqref{eq:MtildetoString} implies that ground-state topological order of $H$, as detected by string order, is equivalent to paramagnetic order in the Kramers-Wannier dual Ising models Eq. \eqref{eq:hedual} and \eqref{eq:hmdual}, whose effective transverse-field strengths $\tilde{h}_X = h_X^{-1}$ and $\tilde{h}_Z = h_Z^{-1}$ respectively. Meanwhile, Eq. \eqref{eq:MtoComp} implies that ground-state computational order of $H$, as detected by ``loop order'', is equivalent to ferromagnetic order in the Ising models Eq. \eqref{eq:directE} and Eq. \eqref{eq:directB} with effective transverse field strengths $h_X$ and $h_Z$ respectively. It follows that in the thermodynamic limit, the computational and topological phases of $H$ are identical and arise for $|h_X|, \, |h_Z| < 1$ as depicted in Fig. \ref{fig:anyons}.

More abstractly, this discussion reveals that for the anisotropic toric code Hamiltonian $H$, computational and topological order are ``two sides of the same coin'': their respective loop and string order parameters both probe ferromagnetic order, but for small perturbations $|h_X|,\,|h_X| \ll 1$ this ferromagnetic order (or lack thereof) arises in very different models, weakly and strongly coupled Ising chains respectively, which have qualitatively distinct physics but are nevertheless related by a Kramers-Wannier duality.

It is instructive to check numerically that our conclusions about the computational ground-state phase of $H$, which were based on reasoning in the thermodynamic limit, are borne out in finite systems, which are required for MBQC. To this end we simulate the true, finite-size MBQC Hamiltonian Eq. \eqref{eq:tch_mbqc} and its associated loop operators \eqref{eq:tc cop e} and \eqref{eq:tc cop m}. We compare the behaviour of these operators to the finite-system string operators
\begin{equation}
    W[\Gamma_o^e] = \prod_{x=\frac{L_x}{2}}^{L_x-1} Z_{(x+\frac12,y)}, \quad W[\Gamma_o^m] = \prod_{x=\frac{L_x}{2}}^{L_x-1}  X_{(x,y+\frac12)}.
\end{equation}
By the discussion above, we expect finite-size loop and string operators to probe the same phase transition at $h_X= h_Z = \alpha=1$, and to have nonzero ground-state expectation values on opposite sides of this transition. These expectations are confirmed by a density matrix renormalization group (DMRG) simulation of the ground state of $H_{\mathrm{MBQC}}$, as shown in Fig. \ref{fig:dmrg results} (see Appendix \ref{sec:numerics} for technical details).

\begin{figure}
    \centering
    \begin{tabular}{cc}
         \includegraphics[width=0.48\textwidth]{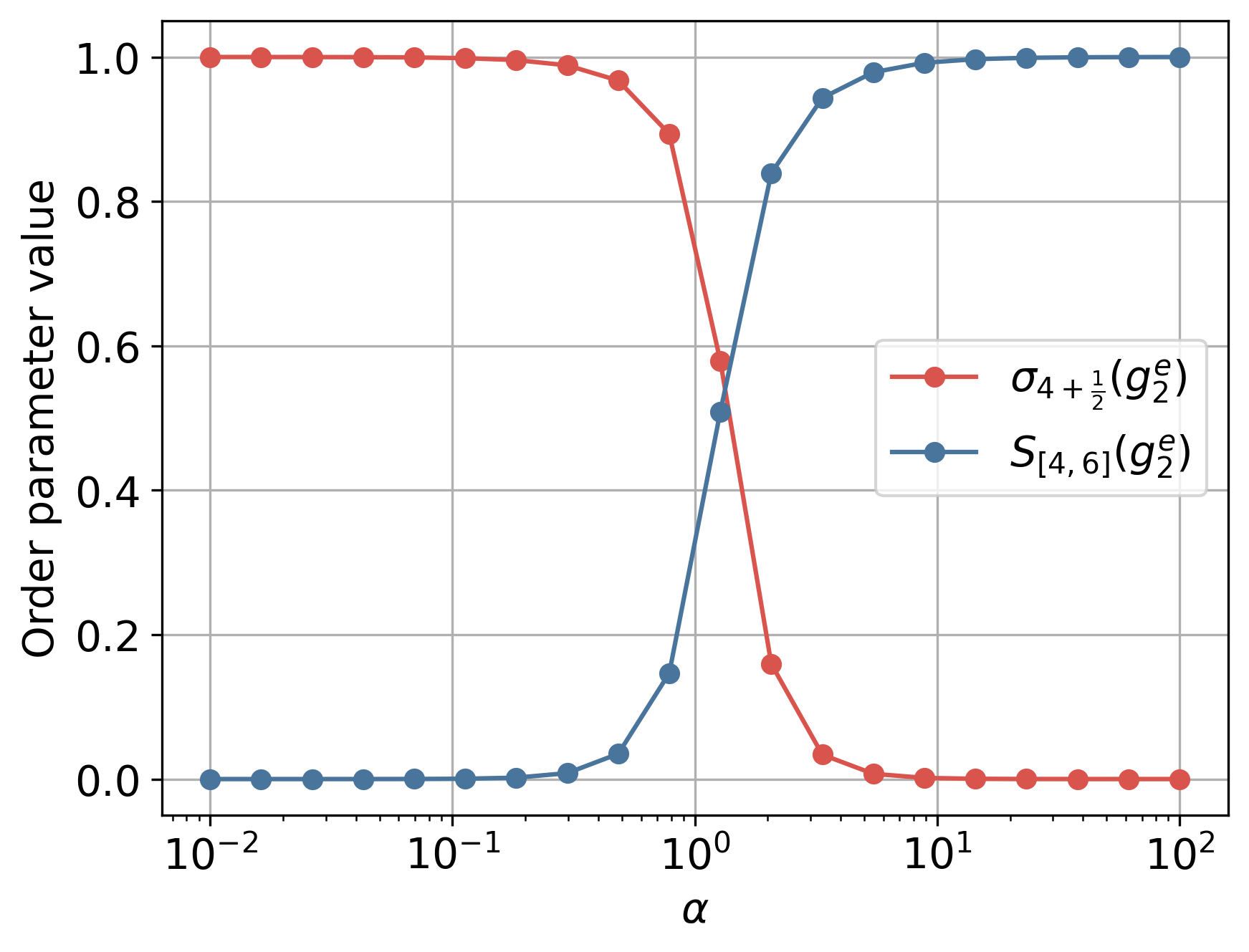} & 
         \includegraphics[width=0.48\textwidth]{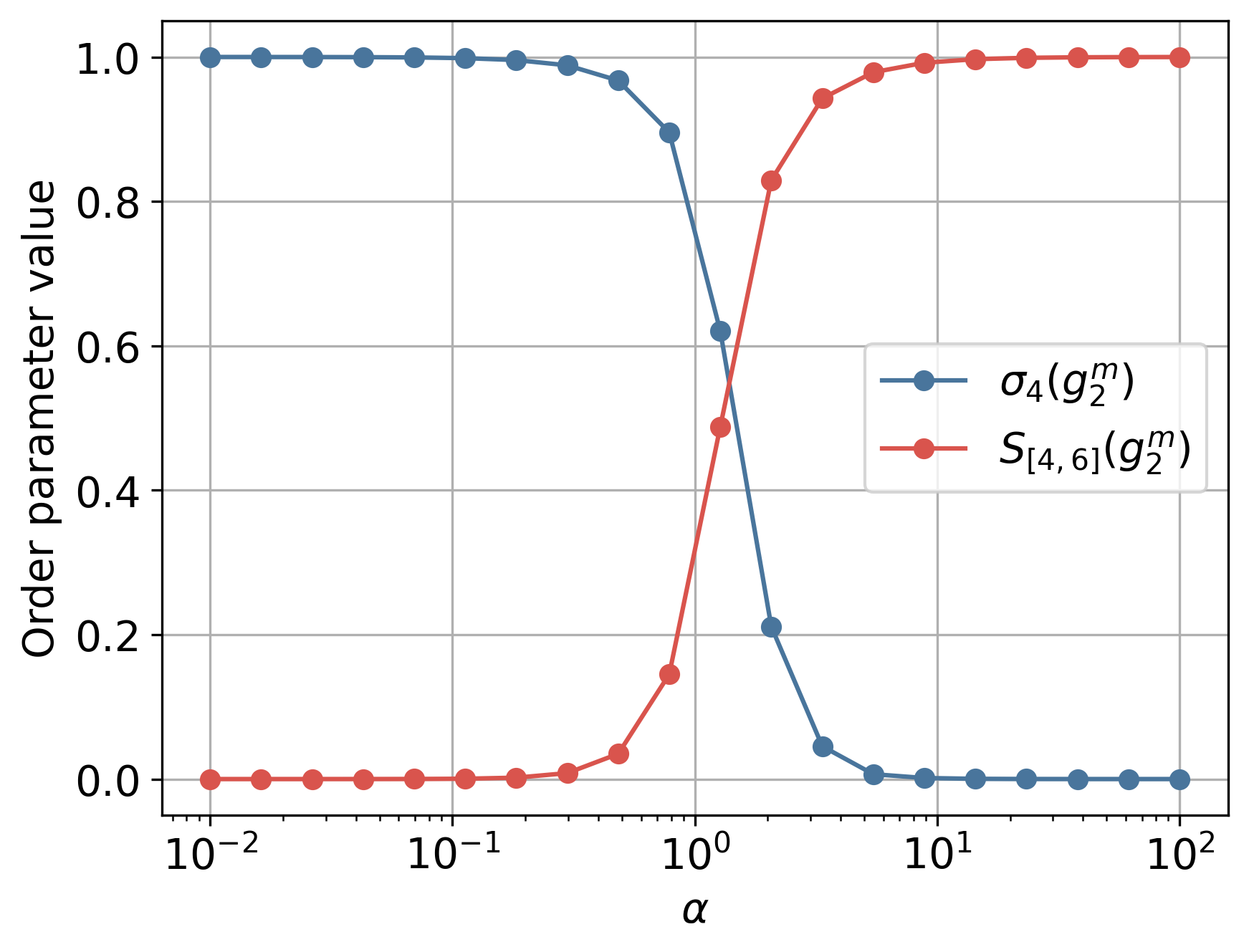} \\
         (a) & (b) 
    \end{tabular}
    \caption{DMRG simulation of (a) electric and (b) magnetic order parameters in the ground state of $H_{\mathrm{MBQC}}$ with $L_x=7, L_y=3$. We observe that the computational order parameters $\sigma_k$ and string order parameters $S_{[a,b]}$ are non-zero on opposite sides of the expected phase transition at $\alpha=1$. For details of the simulation techniques, see Appendix~\ref{sec:numerics}.}
    \label{fig:dmrg results}
\end{figure}

\section{A universal MBQC resource state with SET order}
\label{sec:xz star}

The toric code model described above is well-understood from the perspective of topological order and condensed-matter physics, but from the perspective of MBQC it has only limited computational power. As discussed in Section \ref{sec:tc comp}, the Lie group of gates available for MBQC with the toric code ground state is $\SO(2L_y)$, which is not a universal group of gates -- in fact, quantum computation restricted to operations from $SO(2L_y)$ is efficiently classically simulable~\cite{valiant2001quantum,Terhal_2002}. In other words, while the computational power of the toric code is non-trivial, it is far from optimal. Given this, one might ask whether topologically ordered states limit the computational power available for MBQC.

We answer this question in the negative by presenting a state that is computationally universal for MBQC, and also topologically ordered. Our example is the ground state of a stabilizer Hamiltonian on a square 2D lattice with qubits on the vertices, and dimensions $L_x\times L_y$. We call this model the $XZ$-star model, because the Hamiltonian is a sum of star-shaped operators
\begin{equation}
    H_{*} = -\sum_x\sum_{y\text{ even}} C_{x,y} -\sum_x\sum_{y\text{ odd}} D_{x,y}
\end{equation}
where
\begin{equation}
    C_{x,y} = X_{x-1,y}X_{x,y}X_{x+1,y}X_{x,y-1}X_{x,y+1}, \quad D_{x,y} = Z_{x-1,y}Z_{x,y}Z_{x+1,y}Z_{x,y-1}Z_{x,y+1}.
\end{equation}

\subsection{Topological order}
\label{sec:xz topo}
The elementary excitations of $H_{*}$ come in pairs and are created by string operators $W[\Gamma^e]$ and $W[\Gamma^o]$ acting along strings $\Gamma^e$ ($\Gamma^o$) that start and end on sites with even (odd) $y$ coordinate. Given a string $\Gamma^e$ starting at site $(x,y)$, $W[\Gamma^e]$ acts in a 3-periodic pattern $Z_{x,y}Z_{x+1,y}I_{x+2,y}...$ when moving horizontally and on every second site, i.e. $Z_{x,y}I_{x,y+1}...$ when moving vertically. $W[\Gamma^o]$ acts in the same pattern but with $X$ instead of $Z$. Thus the excitations can only move by hopping 3 sites horizontally or 2 sites vertically. 

There are 4 independent types of excitations in the model: $e$ ($\bar{e}$) created in pairs by $W[\Gamma^e]$ starting at $x=0$ ($x=1$) mod $3$, respectively, and $m$ ($\bar{m}$) created in pairs by $W[\Gamma^o]$ starting at $x=1$ ($x=0$) mod $3$. From the commutation relations of the string operators we can deduce that $\{e, \bar{e}, m, \bar{m}\}$ reproduce the anyon theory of two copies of the toric code, therefore we conclude that there is a finite-depth quantum circuit that maps $H_{*}$ onto two copies of the toric code Hamiltonian \cite{Haah2021}. However, such a circuit generically has a non-trivial interplay with the kinds of logical operations that can be implemented with single-site measurements, so this conclusion does not prohibit computational universality.

\subsection{Subsystem symmetries and SSET order}
\label{sec:xz sset}

The $XZ$-star model admits a set of subsystem symmetries that are similar to those in our toric code model:
\begin{equation}
\label{eq:xz symms}
\begin{aligned}
    &U(g_{110}^y) = \prod_{x=0,1, \, \mathrm{mod} \, 3} C_{x,y}, \quad &&U(g_{011}^y) = \prod_{x=1,2, \, \mathrm{mod} \, 3} C_{x,y} \\
    &U(\bar{g}_{110}^y) = \prod_{x=0,1, \, \mathrm{mod} \, 3} D_{x,y}, \quad &&U(\bar{g}_{011}^y) = \prod_{x=1,2, \, \mathrm{mod} \, 3} D_{x,y}.
\end{aligned}
\end{equation}
These are parallel lines of $X$ or $Z$ in a 3-periodic pattern. $U(g)$ acts on odd rows and $U(\bar{g})$ acts on even rows. Since we have yet to perturb the model, $H_{*}$ actually has a much larger set of symmetries. Here, we have chosen a minimal set of symmetries that are needed to protect the universal computational power of the model, as we demonstrate shortly. 

These symmetries admit a set of global relations, namely,
\begin{equation}
    \begin{aligned}
        &\prod_{y \text{ even}} U(g^y_{110}) = 1,\quad && \prod_{y \text{ even}} U(g^y_{011}) = 1 \\
        &\prod_{y \text{ odd}} U(\bar{g}^y_{110}) = 1,\quad && \prod_{y \text{ odd}} U(\bar{g}^y_{011}) = 1
    \end{aligned}
\end{equation}
As for the toric code, we determine the SSET phase that $H_{*}$ belongs to by calculating the function $\phi$ that identifies the anyon type of the boundary Wilson loop that arises upon restricting these relations to rectangular regions. Skipping the details, we find that
\begin{equation}
    \begin{aligned}
        &\phi\left(\prod_{y \text{ even}} U(g^y_{110})\right)= m,\quad &&\phi\left( \prod_{y \text{ even}} U(g^y_{011})\right)= \bar{m} \\
        &\phi\left(\prod_{y \text{ odd}} U(\bar{g}^y_{110})\right)= e,\quad && \phi\left(\prod_{y \text{ odd}} U(\bar{g}^y_{011})\right)= \bar{e}
    \end{aligned}
\end{equation}
This is the same fractionalization pattern as two decoupled toric codes, with each individual toric code equipped with the subsystem symmetry described in Sec.~\ref{sec:tc comp}. Such a decoupled stack of toric codes is not universal for MBQC, since we have shown that a single toric code is not universal and single-site measurements cannot couple the two states. Thus, even though $H_{*}$ and two toric codes share the same topological order and the same kind of SSET order, they do not have the same computational power. This matches similar observations that have been made for SPT phases \cite{Stephen2019subsystem}. The explanation lies in the symmetry representation, which differs between $H_{*}$ and two decoupled toric codes. In the context of MBQC, the microscopic representation of the symmetry of the lattice is an important ingredient due to the restriction to single-site measurements. In the context of phases of matter, the representation often plays a secondary role since it can change under real-space renormalization or local quantum circuits. We conclude that computational universality of the ground states of $H_{*}$ lies not in their topological or SSET order, but rather in the particular way that the two contituent toric codes have been intertwined. 


\subsection{Computational order}
\label{sec:xz comp}

For MBQC we will consider open boundary conditions, with $x=0,1,\ldots,L_x$, $y=1,2,\ldots,L_y$, and $L_x$ divisible by 3. As in the toric code example, we modify the stabilizers at the boundary by cutting them off. This causes the subsystem symmetries to act projectively at the boundaries. Specifically, we have the projective representations
\begin{equation}
\label{eq:xz proj}
\begin{aligned}
    &V_L(g_{110}^y) = X_{0,y} 
    && V_R(g_{110}^y) = X_{L_x,y} \\
    &V_L(\bar{g}_{110}^y) = Z_{0,y}
    &&V_R(\bar{g}_{110}^y) = Z_{L_x,y} \\
    &V_L(g_{011}^y) = X_{0,y-1}X_{0,y+1} 
    && V_R(g_{011}^y) = X_{L_x,y-1}X_{L_x,y+1} \\
    &V_L(\bar{g}_{011}^y) = Z_{0,y-1}Z_{0,y+1}
    && V_R(\bar{g}_{011}^y) = Z_{L_x,y-1}Z_{L_x,y+1} \\
    &V_L(g_{101}^y) = X_{0,y-1}X_{0,y}X_{0,y+1}
    &&V_R(g_{101}^y) = X_{L_x,y-1}X_{L_x,y}X_{L_x,y+1} \\
    &V_L(\bar{g}_{101}^y) = X_{0,y-1}X_{0,y}X_{0,y+1} 
    && V_R(\bar{g}_{101}^y) = Z_{L_x,y-1}Z_{L_x,y}Z_{L_x,y+1} .
\end{aligned}
\end{equation}
A maximal abelian subgroup of $V_R(G)$ is generated by $\{V_R(\bar{g}_{110}^y), V_R(\bar{g}_{011}^y):y=1,2,\ldots,L_y\}$. Therefore, we define $H = \expval{\bar{g}_{110}^y, \bar{g}_{011}^y:y=1,2,\ldots,L_y}$
and choose the logical observables $T(G)$ according to \eqref{eq:logi_obs}. The rank of $H$ is $L_y-1$, therefore we have $L_y-1$ logical qubits.

To determine the Lie group of gates for MBQC, we need to construct the operators $R_k(g)$. As in the toric code example, we do this by solving \eqref{eq:localizability condition} for every site $(k,y)$ in the bulk. This yields
\begin{equation}
\label{eq:xz R}
\begin{aligned}
    &R_k(g_{110}^y) = \prod_{x>k, x=0,1, \, \mathrm{mod} \, 3} C_{x,y}, \quad
    &&R_k(\bar{g}_{110}^y) = \prod_{x>k, x=0,1, \, \mathrm{mod} \, 3} D_{x,y}, \\
    &R_k(g_{011}^y) = \prod_{x>k, x=1,2, \, \mathrm{mod} \, 3} C_{x,y},\quad
    &&R_k(\bar{g}_{011}^y) = \prod_{x>k, x=1,2, \, \mathrm{mod} \, 3} D_{x,y}, \\
    &R_k(g_{101}^y) = \prod_{x>k, x=0,2, \, \mathrm{mod} \, 3} C_{x,y},\quad
    &&R_k(\bar{g}_{101}^y) = \prod_{x>k, x=0,2, \, \mathrm{mod} \, 3} D_{x,y}.
\end{aligned}
\end{equation}
Therefore, we have $\mathcal{G} = \{g_{110}^y, \bar{g}_{110}^y, g_{011}^y, \bar{g}_{011}^y, g_{101}^y, \bar{g}_{101}^y:y=1,...,L_y\}$. 

The Lie algebra of MBQC gates is generated by the projective representations $V_R(g), g\in \mathcal{G}$, which are precisely the representations given in \eqref{eq:xz proj}. To determine which Lie algebra is generated, we note that any two sets of $n$-qubit Pauli operators with the same commutation relations generate the same Lie algebra. Thus we need only find a set of Paulis with the same commutation relations as \eqref{eq:xz proj} and such that the Lie algebra that they generate is known. One set satisfying these conditions is the set of terms from the Heisenberg chain Hamiltonian: $\{X_iX_{i+1}, Y_iY_{i+1}, Z_iZ_{i+1}:i=1,...,L_y-1\}$. The algebra generated by this set -- and therefore by $V_R(\mathcal{G})$ -- is shown in Theorem IV.1 of~\cite{wiersema2023classificationdynamicalliealgebras} to be
\begin{equation}
    \mathcal{A} =
    \begin{cases}
        \su(2^{L_y-1}) & L_y \text{ odd} \\
        \su(2^{L_y-2})^{\oplus 4} & L_y \text{ even}.
    \end{cases}
\end{equation}
Thus the Lie group $\mathcal{L}$ of MBQC gates is $\SU(2^{L_y-1})$ for $L_y$ odd and $\SU(2^{L_y-2})$ for $L_y$ even, which are universal gate sets for $L_y-1$ and $L_y-2$ qubits, respectively.

\section{Discussion}
\label{sec:disc}

Our formalism for computational order in 2D systems allows us to describe the first examples of computational phases of matter with topological order. Our results take the notion of computational order beyond SPT systems, which were previously the only systems known to exhibit computational phases. In other words, computational order is not simply a reformulation of SPT order.

However, computational order also fails to correspond neatly to SET order or SSET order, as illustrated by our $XZ$-star model. Instead, there appears to be a more complicated interplay between computational power and the action of the symmetries that protect it. It appears to us that understanding the precise nature of this relationship between symmetry, computational order, and physical order is an important direction for future work.

Subsystem symmetry seems particularly significant from this perspective because it appears in nearly all examples of non-trivial computational order, including those presented here. Now that we can characterize the computational power of long-range entangled systems, it would be interesting to study the models described in Ref.~\cite{Stephen2020,Stephen2022} from the perspective of MBQC. For example, for the models described in Ref.~\cite{Stephen2022}, there are simple external fields that drive the model to a trivial SSET phase without destroying the topological order. Studying this transition would allow the roles of topological order and subsystem symmetry enrichment in MBQC universality to be more clearly distinguished.

This generalization of utility for MBQC from SPT phases to topological phases mirrors a similar generalization that was obtained recently in the setting of quantum nonlocal games~\cite{brassard2005quantum}. The first works on playing nonlocal games with ground states of local Hamiltonians found that a robust quantum advantage for such games could be obtained in SPT~\cite{daniel2021quantum,bulchandani2023playing} and conventional symmetry-breaking~\cite{bulchandani2023playing} phases of matter, but an apparent lack of robustness in topological phases, for example, for nonlocal games played in the topological phase of the toric code perturbed by bond-isotropic magnetic fields~\cite{bulchandani2023playing,bulchandani2023multiplayer}. This lack of robustness was remedied by a game that could be played with local~\cite{hart2024playing}, rather than nonlocal~\cite{bulchandani2023multiplayer}, stabilizers of the toric code. Thus, although quantum advantage for nonlocal games appears to be less demanding than fully-fledged MBQC (no subsystem symmetry is required), the properties that a quantum phase of matter must possess in order to yield a robust quantum advantage for nonlocal games (in the sense of e.g. Ref.~\cite{lin2023quantum}) remain obscure even in this simpler setting.

More broadly, our results highlight the significant overlap between objects of interest for MBQC---symmetries, order parameters, and logical observables---and objects of interest to condensed matter physics. For the examples studied in this paper, especially the anisotropic toric code, we were able to make these connections quantitatively precise, but there appears to be plenty of scope for developing broader and more systematic connections between computational power for MBQC and the kinds of quantum states that arise naturally in condensed matter physics. 

\section{Acknowledgments}
We thank Arnab Adhikary, Ehud Altman, Michael Knap, Tobias J. Osborne and Frank Pollmann for helpful discussions. PH is funded by the National Science and Engineering Research Council of Canada (NSERC). VBB thanks the Simons Institute for the Theory of Computing for their hospitality during part of the completion of this work. YJ acknowledges funding by the Ministry of Science and Culture of Lower Saxony through \textit{Quantum Valley Lower Saxony Q1 }(QVLS-Q1). DTS acknowledges support from the Simons Collaboration on Ultra-Quantum Matter, which is a grant from the Simons Foundation (651440). RR is funded by the Humboldt Foundation.

\bibliography{bibl}

\newpage
\appendix

\section{Proof of Theorem \ref{thm:power}}
\label{sec:power proof}

Recall Theorem \ref{thm:power}:

\power* 

This proof is a generalization of the proof given in~\cite{Raussendorf2023measurementbased} to 2D systems. We will make use of the operators
\begin{equation}
\label{eq:l def}
    L_k(g) = T(g)R_k(g)
\end{equation}
which act non-trivially on the region to the left of site $k$. Before we continue, it will be useful to establish some commutation relations between the various operators defined so far.

\begin{lemma} \label{lemma:comm rel}
    The symmetry representations $U(G)$, $T(G)$, $R_k(\mathcal{G}_k)$, and $L_k(\mathcal{G}_k)$ have the following commutation relations:
    \begin{align}
        &\comm{U(g)}{R_k(g_k)} = 0 \quad g\in G, \; g_k\in \mathcal{G}, \;k\in B \label{eq:comm u r} \\
        &\comm{T(g)}{R_k(g_k)} = 0 \quad g\in G, \; g_k\in\mathcal{G}, \;k\in B \label{eq:comm t r} \\
        &T_k(g')T(g) = (-1)^{\kappa(g,g')} T(g)T_k(g'), \quad g, g'\in G \label{eq:comm t t}  \\
        &L_k(g_k)T(g) = (-1)^{\kappa(g,g_k)} T(g)L_k(g_k), \quad g\in G, \; g_k\in \mathcal{G}, \;k\in B\label{eq:comm t l} 
    \end{align}
\end{lemma}
\begin{proof}
    The operators $R_k(g)$ are already constructed such that \eqref{eq:comm u r} holds. $T(g)$ differ from $U(g)$ only on the left boundary, so \eqref{eq:comm t r} follows from \eqref{eq:comm u r}. To prove \eqref{eq:comm t t}, recall that $T(g)$ are constructed such that they all commute on the left boundary, so they inherit their commutation relations from the right boundary, which are precisely the relations encoded by $\kappa(g,g')$. Finally, for \eqref{eq:comm t l}, we can apply the definition \eqref{eq:l def} and \eqref{eq:comm t r} to show that
    \begin{equation}
        L_k(g_k)T(g) = T(g_k)R_k(g_k)T(g) = (-1)^{\kappa(g,g_k)}T(g)T(g_k)R_k(g_k) = (-1)^{\kappa(g,g_k)}T(g)L_k(g_k).
    \end{equation}
\end{proof}


The rotation operators obtained by exponentiating $L_k(g_k)$ are
\begin{equation}
    M_{k}(g_{k}, \theta_{k}) := \exp(-i\frac{\theta_{k}}{2} L_{k}(g_{k})) \qquad g_{k} \in \mathcal{G}_{k}.
\end{equation}
For each rotation there is a corresponding logical rotation
\begin{equation}
    M_{\log,k}(g_k, \theta_k) = \exp(-i\frac{\theta_k}{2}T(g_k))
\end{equation}
and a logical CPTP map
\begin{equation}
    \mathcal{M}_k(g_k, \theta_k) = \frac{1+\sigma_k(g_k)}{2}\qty[M_{\log,k}(\theta_k)] + \frac{1-\sigma_k(g_k)}{2}\qty[M^\dagger_{\log,k}(\theta_k)]
\end{equation}
In the following lemma, we will show that conjugation by $M_k$ is equivalent to the logical CPTP map $\mathcal{M}_k$, in the sense that it reproduces the same measurement statistics for all logical observables.

\begin{lemma} \label{lemma:maps}
    Consider a state $\ket{\Phi}$ satisfying conditions \ref{item:finite spread}, \ref{item:good stabs} and \ref{item:g symmetric} of Theorem \ref{thm:power}, and a sequence of sites $\{k_1, , ..., k_f\}$ such that the horizontal distance between site $k_i$ and $k_{i-1}$ is at least $2\Delta+d$ for all $i$.
    Then the following equation holds:
    \begin{equation}
        \bra{\Phi} M^\dagger_{k_1} ... M^\dagger_{k_f} T(g) M_{k_f} ... M_{k_1} \ket{\Phi} = \Tr\qty[T(g)\mathcal{M}_{k_f} ... \mathcal{M}_{k_1} (\dyad{\Phi}{\Phi})]. 
    \end{equation}
\end{lemma}
\begin{proof}
    We will prove a single step of the evolution and the rest follows by induction. Let $\ket{\Phi(k_i)} = M_{k_i} ... M_{k_1} \ket{\Phi}$, and consider the expectation value 
    \begin{equation}
        \bra{\Phi(k_{i-1})} M_{k_i}^\dagger T(g) M_{k_i} \ket{\Phi(k_{i-1})}
    \end{equation}
    If $T(g_{k_i})$ commutes with $T(g)$, then by Lemma \ref{lemma:comm rel} so does $L_{k}(g_{k_i})$ and we have
    \begin{equation} \label{eq:TL commute}
        \bra{\Phi(k_{i-1})} M_{k_i}^\dagger T(g) M_{k_i} \ket{\Phi(k_{i-1})} = \bra{\Phi(k_{i-1})} T(g) \ket{\Phi(k_{i-1})}.
    \end{equation}
    If $T(g_{k_i})$ does not commute with $T(g)$ then it must anticommute, in which case
    \begin{equation} \label{eq:TL anticommute}
        \begin{aligned}
            &\bra{\Phi(k_{i-1})} M_{k_i}^\dagger T(g) M_{k_i} \ket{\Phi(k_{i-1})} = \\
            & \qquad \qquad \cos(\theta_{k_i})\bra{\Phi(k_{i-1})} T(g) \ket{\Phi(k_{i-1})} - i\sin(\theta_{k_i}) \bra{\Phi(k_{i-1})} T(g) L_{k_i}(g_{k_i}) \ket{\Phi(k_{i-1})} 
        \end{aligned}
    \end{equation}
    For the second term, we have 
    \begin{equation} \label{eq:pull out R}
        \begin{aligned}
            \bra{\Phi(k_{i-1})} T(g) L_{k_i}(g_{k_i}) \ket{\Phi(k_{i-1})} &= \bra{\Phi(k_{i-1})} T(g) T(g_{k_i}) R_{k_i}(g_{k_i}) \ket{\Phi(k_{i-1})} \\
            &= \bra{\Phi} M^\dagger_{k_1} ... M^\dagger_{k_{i-1}} T(g) T(g_{k_i}) M_{k_{i-1}} ... M_{k_1} R_{k_i}(g_{k_i}) \ket{\Phi} \\ 
            &= \bra{\Phi} U^\dagger(g g_{k_i}) M^\dagger_{k_1} ... M^\dagger_{k_{i-1}} T(g) T(g_{k_i}) M_{k_{i-1}} ... M_{k_1} R_{k_i}(g_{k_i}) \ket{\Phi} \\
            &= \bra{\Phi} U^\dagger(g g_{k_i}) M^\dagger_{k_1} ... M^\dagger_{k_{i-1}} T(g) T(g_{k_i}) M_{k_{i-1}} ... M_{k_1} \ket{\Phi}\bra{\Phi}R_{k_i}(g_{k_i}) \ket{\Phi} \\
            &= \sigma_{k_i}(g_{k_i}) \bra{\Phi(k_{i-1})} T(g) T(g_{k_i}) \ket{\Phi(k_{i-1})}, 
        \end{aligned}
    \end{equation}
    In the first line we used \eqref{eq:l def} and in the second line we moved $R_{k_i}$ past $M_{k_{i-1}} ... M_{k_1}$ because they have disjoint support. In the third line we pull out a global symmetry $U(g g_{k_i})$ and note that
    \begin{equation}
        \supp(U^\dagger(g g_{k_i}) M^\dagger_{k_1} ... M^\dagger_{k_{i-1}} T(g) T(g_{k_i}) M_{k_{i-1}} ... M_{k_1}) \subseteq \{\leq k_{i-1}\} 
    \end{equation} 
    while $\supp(R_{k_i}(g_{k_i})) \subseteq \{\geq k_i\}$, which allows us to apply Lemma \ref{lemma:decorr} in the fourth line.

    Combining \eqref{eq:TL commute}-\eqref{eq:pull out R}, we have 
    \begin{equation} 
        \begin{aligned}
            \bra{\Phi(k_{i-1})} M_{k_i}^\dagger T(g) M_{k_i} \ket{\Phi(k_{i-1})} &= \frac{1+\sigma_{k_i}(g_{k_i})}{2} \bra{\Phi(k_{i-1})} M_{\log,k_i}^\dagger T(g) M_{\log,k_i} \ket{\Phi(k_{i-1})} \\ 
            &\qquad + \frac{1-\sigma_{k_i}(g_{k_i})}{2} \bra{\Phi(k_{i-1})} M_{\log,k_i} T(g) M_{\log,k_i}^\dagger \ket{\Phi(k_{i-1})} \\
            &= \bra{\Phi(k_{i-1})} \mathcal{M}_{k_i}^\dagger (T(g)) \ket{\Phi(k_{i-1})} \\
            &= \Tr[T(g) \mathcal{M}_{k_i} (\ket{\Phi(k_{i-1})}\bra{\Phi(k_{i-1})}) ].
        \end{aligned}
    \end{equation}
\end{proof}
\begin{lemma} \label{lemma:local meas}
    Consider a state $\ket{\Phi}$ and sequence of sites $\{k_1,...,k_f\}$ as in Lemma \ref{lemma:maps}. The expectation value
    \begin{equation}
        \bra{\Phi} M^\dagger_{k_1} ... M^\dagger_{k_d} T(h) M_{k_d} ... M_{k_1} \ket{\Phi}, \qquad h\in H
    \end{equation}
    can be reproduced by adaptive single-qubit measurements on $\ket{\Phi}$ for any subgroup $H\subset G$ such that $V_R(H)$ is abelian.
\end{lemma}
\begin{proof}
    Denote
    \begin{equation} 
        U_L(g) = 
        \begin{cases}
        V_L(g), & g \in H \\
        I_L, & g \in G \setminus H
        \end{cases}
    \end{equation}
    and let $U^B_{<k}(g)$ the restriction of $U_B(g)$ to $\{<k\}$. Note that this is different than $U_{<k}(g) = V_L(g) \otimes U^B_{<k}(g)$.
    
    Let $M_{\leq k_i} = M_{k_i}...M_{k_1}$, and define the intermediate logical observables 
    \begin{equation}
    \label{eq:inter obs}
    \begin{aligned}
        &T_{0}(g) = U_L(g) \\
        &T_{\leq k_i}(g) = M_{\leq k_i}^\dagger U_L(g) \otimes U^B_{\leq k_i}(g) M_{\leq k_i}.
    \end{aligned}
    \end{equation}
    Also, let $U^B_{k_i:k_j}(g) = \bigotimes_{x \in \{< k_j\} \setminus \{\leq k_i\}} \hat{u}_x(g)$ be the restriction of $U_B(g)$ to sites whose horizontal position is between $k_i$ and $k_j$, and let $\partial\{< k\}$ be the set of sites on the boundary of $\{< k\}$. The intermediate observables satisfy a recursion relation
    \begin{equation} \label{eq:recursion v1}
        \begin{aligned}
            T_{\leq k_{i+1}}(g) &= M_{\leq k_{i+1}}^\dagger U_L(g) \otimes U^B_{\leq k_{i+1}}(g) M_{\leq k_{i+1}} \\
            &= M_{\leq k_{i+1}}^\dagger \qty( U_L(g) \otimes U^B_{\leq k_i}(g) \otimes U^B_{k_i:k_{i+1}}(g) \otimes \qty(\bigotimes_{j\in\partial \{< k_{i+1}\}} \hat{u}_j(g) )) M_{\leq k_{i+1}} \\
            &= M_{\leq k_i}^\dagger U_L(g) \otimes U^B_{\leq k_i}(g) M_{\leq k_i} \cdot U^B_{k_i:k_{i+1}}(g) \cdot \qty( \bigotimes_{j\in\partial \{< k_{i+1}\} \setminus k_{i+1}} \hat{u}_j(g) ) \cdot M_{\leq k_{i+1}}^\dagger \hat{u}_{k_{i+1}}(g) M_{\leq k_{i+1}} \\
            &= T_{\leq k_i}(g) \cdot U^B_{k_i:k_{i+1}}(g) \cdot \qty( \bigotimes_{j\in\partial \{< k_{i+1}\} \setminus k_{i+1}} \hat{u}_j(g) ) \cdot M_{\leq k_i}^\dagger M_{k_{i+1}}^\dagger M_{\leq k_i} \hat{u}_{k_{i+1}}(g) M_{\leq k_i}^\dagger M_{k_{i+1}} M_{\leq k_i} \\
        \end{aligned}
    \end{equation}
    where in the third and fourth lines we use that $L_{k_i}(g_{k_i})$ -- and therefore $M_{k_i}$ -- commutes with $\hat{u}_j(g)$ for all $j\neq k_i$ and $g\in G$. The product of rotations in the last line simplifies to
    \begin{equation}
        \begin{aligned}
            M_{\leq k_i}^\dagger M_{k_{i+1}}^\dagger M_{\leq k_i} &= M_{\leq k_i}^\dagger \exp(-i\frac{\theta_{k_{i+1}}}{2} L_{k_{i+1}}(g_{k_{i+1}})) M_{\leq k_i} \\
            &=  \exp(-i\frac{\theta_{k_{i+1}}}{2} M_{\leq k_i}^\dagger L_{k_{i+1}}(g_{k_{i+1}})M_{\leq k_i} ) \\
            &= \exp(-i\frac{\theta_{k_{i+1}}}{2} T_{\leq k_i}(g_{k_{i+1}}) \otimes \hat{\beta}_{k_{i+1}}(g_{k_{i+1}}))
        \end{aligned}
    \end{equation}
    We use this to define the observables
    \begin{equation} \label{eq:local_observable}
        O_{\leq k_i}(g, g_{k_i}, \theta_{k_i}) = \exp(i\frac{\theta_{k_i}}{2} T_{\leq k_{i-1}}(g_{k_i})\otimes \hat{\beta}_{k_i}(g_{k_i})) \hat{u}_{k_i}(g) \exp(-i\frac{\theta_{k_i}}{2} T_{\leq k_{i-1}}(g_{k_i})\otimes \hat{\beta}_{k_i}(g_{k_i})).
    \end{equation}
    Then the recursion relation \eqref{eq:recursion v1} becomes
    \begin{equation} \label{eq:recursion v2}
        T_{\leq k_{i+1}}(g) = T_{\leq k_i}(g) \cdot U^B_{k_i:k_{i+1}}(g) \cdot \qty( \bigotimes_{j\in\partial \{< k_{i+1}\} \setminus k_{i+1}} \hat{u}_j(g) ) O_{\leq k_i}(g, g_{k_i}, \theta_{k_i}).
    \end{equation}
    
    Now we show by induction that the measurement outcome $\lambda_{\leq k_i}(g)$ of $T_{\leq k_i}(g)$ can be inferred from single-qubit measurements. For the first step, $T_{0}(g) = U_L(g)$, we have that $\lambda_{0}(g)$ can be simultaneously measured for all $g\in G$ because $U_L(G)$ are products of commuting single-qubit observables. Then we can replace $T_{0}(g)$ by $\lambda_{0}(g)$ in \eqref{eq:local_observable} and \eqref{eq:recursion v2}, yielding
     \begin{equation}
        T_{\leq k_1}(g) = \lambda_{0}(g) \cdot U^B_{<k_1}(g) \cdot \qty( \bigotimes_{j\in\partial \{< k_{1}\} \setminus k_{1}} \hat{u}_j(g) ) O_{k_1}(g, g_{k_1}, \theta_{k_1}),
    \end{equation}
    where
    \begin{equation}
        O_{k_1}(g,g_{k_1},\theta_{k_1}) =  \exp(i\frac{\lambda_{0}(g) \theta_{k_1}}{2} \hat{\beta}_{k_1}(g_{k_1})) \hat{u}_{k_1}(g) \exp(-i\frac{\lambda_{0}(g) \theta_{k_1}}{2} \hat{\beta}_{k_1}(g_{k_1}))
    \end{equation}
    The observables $U^B_{<k_1}(g) \cdot \qty( \bigotimes_{j\in\partial \{< k_{1}\} \setminus k_{1}} \hat{u}_j(g) )$ are a product of commuting single-qubit observables, so they can be measured simultaneously for all $g\in G$. Then note that the outcome of $O_{k_1}(g,g_{k_1},\theta_{k_1})$ for all $g\in G$ can be inferred by measuring the qubit at site $k_1$ in a rotated basis 
    \begin{equation}
        \exp(i\frac{\lambda_{0}(g) \theta_{k_1}}{2} \hat{\beta}_{k_1}(g_{k_1})) \hat{\alpha}_{k_1} \exp(-i\frac{\lambda_{0}(g) \theta_{k_1}}{2} \hat{\beta}_{k_1}(g_{k_1}))
    \end{equation}
    where $\hat{\alpha}_{k_1}\in \{X_{k_1},Y_{k_1},Z_{k_1}\}$ is the Pauli operator at site $k_1$ that commutes with $\hat{u}_{k_1}(G)$. Therefore, we can infer $\lambda_{\leq k_1}(g)$ for all $g\in G$ from single-qubit measurements. Similarly, we can infer $\lambda_{\leq k_{i+1}}(g)$ from $\lambda_{\leq k_i}(g)$.
    
    The induction continues until we reach $\lambda_{\leq k_f}(g)$. To obtain the measurement outcome of
    \begin{equation}
        M^\dagger_{k_1} ... M^\dagger_{k_d} T(g) M_{k_d} ... M_{k_1} = T_{\leq k_f}(g) \otimes U^B_{>k_d}(g) \otimes V_R(g),
    \end{equation}
    we need to measure $U^B_{>k_d}(g)$ and $V_R(g)$. $U^B_{>k_d}(g)$ is a linear representation, so it can be simultaneously measured for all $g\in G$, but $V_R(g)$ is a projective representation, thus we can only simultaneously measure a subgroup $H\subset G$ such that $V_R(H)$ is abelian.
\end{proof}

\begin{lemma}
\label{lemma:efficiency}
    Given a state $\ket{\Phi}$ satisfying conditions \ref{item:finite spread}, \ref{item:good stabs} and \ref{item:g symmetric} of Theorem \ref{thm:power}, MBQC can implement the logical unitaries $M_{\log,k}(g_k, \theta_k)$ with error $\epsilon$ under the diamond norm with the properties stated in Lemma 12 of Ref.~\cite{aharonov1998quantum} given by
    \begin{equation}
        \epsilon = \frac{\theta_k^2}{N}\frac{1-\sigma_k(g_k)^2}{\sigma_k(g_k)^2}
    \end{equation}
\end{lemma}

\begin{proof}
    We refer to Appendix B of~\cite{Raussendorf2023measurementbased} where it is proven that a sequence of $N$ logical CPTP maps $\mathcal{M}_k(g_k, \theta_k/N)$ approximate the unitary $M_{\log,k}(g_k, \theta_k)$ to error $\epsilon$.
\end{proof}

\begin{proof}[Proof of Theorem \ref{thm:power}]
With Lemma \ref{lemma:maps} we showed that the expectation values of the logical observables $T(g)$ after conjugation by $M_{k_f} ... M_{k_1}$ are equal to their expectation values after applying the logical CPTP maps $\mathcal{M}_{k_f} ... \mathcal{M}_{k_1}$. With Lemma \ref{lemma:local meas} we showed that these expectation values can be inferred by local measurements. With Lemma \ref{lemma:efficiency} we established that a sequence of $N$ logical CPTP maps with rotation angle $\theta/N$ can approximate a logical unitary $M_{\log,k}(g_k,\theta_k)$ to arbitrary precision as long as the corresponding order parameter $\sigma_k(g_k)$ is non-zero. 

The logical unitaries that are available for MBQC are rotations about axes defined by the logical observables $T(\mathcal{G})$. These logical observables derive their commutation relations from $V_R(\mathcal{G})$, thus the Lie algebra they generate under commutators and linear combinations is isomorphic to the Lie algebra generated by $V_R(\mathcal{G})$. The same statement holds at the level of Lie groups after applying the exponential map. This establishes item \ref{thm_item:logi} of Theorem \ref{thm:power}. To establish item \ref{thm_item:init}, we note that the initial state of the logical observables is given by Lemma \ref{lemma:logi_init}, and Lemma \ref{lemma:local meas} tells us that we can read out $T(H)$ by local measurement. 
\end{proof}

\section{Details of numerical simulations}
\label{sec:numerics}
In this section, we summarize the numerical method used to generate Fig. \ref{fig:dmrg results}. We use White's \textit{Density Matrix Renormalization Group} (DMRG) \cite{whiteDensityMatrixFormulation1992, schollwockDensitymatrixRenormalizationGroup2011}, a classical numerical technique that provides a variational representation of a ground state $\ket{\psi}$. This ground state is expressed as a \textit{Matrix Product State} (MPS), which is structured as follows \cite{fannesFinitelyCorrelatedPure1994, verstraeteMatrixProductStates2006}:
\begin{equation}
\ket{\psi} = \sum_{\substack{\sigma_1 \dots \sigma_N \\ \alpha_1, \dots, \alpha_{N-1}}} T^{\sigma_1}_{\alpha_0, \alpha_1} \cdots T^{\sigma_i}_{\alpha_{i-1}, \alpha_i} \cdots T^{\sigma_N}_{\alpha_{N-1}, \alpha_N} \ket{\sigma_1 \dots \sigma_N},
\end{equation}
where each $T^{\sigma_i}_{\alpha_{i-1}, \alpha_i}$ is a rank-3 tensor. The indices $\alpha_{i-1}$ and $\alpha_i$ span the \textit{bond dimension} $\chi_i$ at site $i$, with the maximum bond dimension denoted $\chi_{\text{max}}$, which controls the accuracy of the MPS. Open boundary conditions are imposed by setting $\chi_0 = \chi_N = 1$.

DMRG is a powerful algorithm that optimizes the variational parameters of the MPS, specifically the rank-3 tensors $T^{\sigma_i}_{\alpha_{i-1}, \alpha_i}$. The algorithm proceeds through a process of sequential quadratic variational optimizations:
\begin{equation}
E_{gs}=\min_{T^{\sigma_i}_{\alpha_{i-1}, \alpha_i}} \langle \psi | \hat{H} | \psi \rangle,
\end{equation}
where the optimization is carried out for each site $i$ in the system and continues until convergence. 

In our simulations, we begin by expressing the Hamiltonian $\hat{H}_{\mathrm{MBQC}}$ as a matrix product operator
\begin{equation}
\hat{H}_{\mathrm{MBQC}} = \sum_{\substack{\sigma_1 \dots \sigma_N \\ \sigma'_1 \dots \sigma'_N}} M^{\sigma_1 \sigma'_1} \cdots M^{\sigma_i \sigma'_i} \cdots M^{\sigma_N \sigma'_N} \ket{\sigma_1 \dots \sigma_N} \bra{\sigma'_1 \dots \sigma'_N},
\end{equation}
where each tensor $M^{\sigma_i \sigma'_i}$ is a $D \times D$ matrix, with $D$ representing the bond dimension between sites $i$ and $i+1$. The states $\ket{\sigma_i}$ and $\ket{\sigma'_i}$ refer to the local basis states at site $i$. To extend the MPS algorithms to two dimensions, one can label the sites of the lattice in a snake-like pattern that spans all the sites vertically or horizontally. In this work, we use the pattern depicted in Fig.~\ref{fig:2D MPS}, in which the direction of the arrow shows the labeling order. By using this convention, one can write the MPO representation of the Hamiltonian or any other observable.

\begin{figure}
    \centering
    \includegraphics[width=0.45\columnwidth]{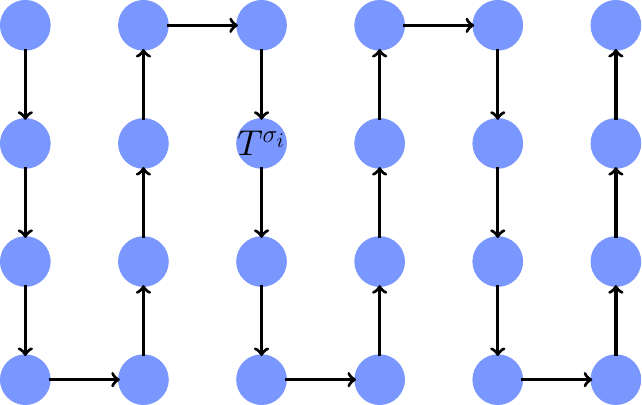}
    \caption{Mapping a two-dimensional lattice to a matrix product state form by labeling each site in a snake-like path through the lattice and representing each site with its corresponding rank-3 tensor.}
    \label{fig:2D MPS}
\end{figure}

To find the MPO representation of the Hamiltonian in the form Eq.~\ref{eq:tch_mbqc}, we use a matrix product diagram based on a finite state machine \cite{crosswhite2008, Keller2015MPO}. Depending on the structure of the Hamiltonian, different approaches are required to determine the MPO representation. When the Hamiltonian follows a regular structure, such as nearest-neighbor interactions, the MPO typically shows uniformity throughout the system. However, for Hamiltonians with arbitrary structures, a matrix product diagram-based representation is necessary~\cite{crosswhite2008, Keller2015MPO, mcculloch_density-matrix_2007, Ren2020}.

The parameters of our DMRG simulation are as follows: maximum bond dimension 200, energy convergence threshold $10^{-5}$, negligible singular value cutoff $10^{-32}$, and a maximum of 20 sweeps. As $\alpha$ becomes much larger than one, the ground state becomes exponentially close to degeneracy because of the lack of magnetic fields at the left and right boundaries. Therefore, instead of setting these fields to zero, we set them to $10^{-3}$, which is always at least an order of magnitude smaller than any other term in the Hamiltonian. 
The raw data for Fig. \ref{fig:dmrg results} and code to reproduce our results is available from the authors upon reasonable request.

\end{document}